\documentclass{article}

\usepackage[T1]{fontenc}
\usepackage{lmodern}
\usepackage{microtype}

\usepackage{graphicx}
\usepackage[round]{natbib}

\usepackage{caption}

\usepackage{amsmath, amsthm, amssymb, amsfonts}
\newtheorem{theorem}{Theorem}[section]
\newtheorem{lemma}{Lemma}[section]
\newtheorem{definition}{Definition}[section]

\numberwithin{equation}{section}

\newcommand{\margin}{30mm}
\usepackage[a4paper, left=\margin, right=\margin, top=\margin, bottom=\margin]{geometry}

\usepackage{bbm}

\usepackage{xcolor}
\definecolor{bodytext}{HTML}{2F2C3D}
\definecolor{jade}{HTML}{01848B}
\definecolor{silver}{HTML}{73787A}
\definecolor{eggplant}{HTML}{9061C4}

\usepackage[
    bookmarks=true,
    unicode=true,
    pdftoolbar=false,
    pdfmenubar=false,
    pdffitwindow=true,
    pdfstartview={FitV},
    pdftitle={Fast-excursion limit of the Heston model},
    pdfauthor={Ryan McCrickerd},
    pdfkeywords={fast-reversion limit; random closed sets; barrier options; excursion risk},
    pdflang={en-GB},
    pdfdisplaydoctitle=true,
    bookmarksnumbered=true,
    pdfnewwindow=true,
    colorlinks=true,
    linkcolor=jade,       
    citecolor=eggplant,   
    urlcolor=silver,      
    filecolor=silver      
]{hyperref}

\usepackage{orcidlink}
\usepackage{doi}

\newcommand{\bb}[1]{{\mathbb{#1}}}
\newcommand{\cc}[1]{{\mathcal{#1}}}
\newcommand{\rr}[1]{{\mathrm{#1}}}

\makeatletter
\newcommand{\nestappendixtoc}{%
    \let\l@section\l@subsection
    \let\l@subsection\@gobbletwo
}
\makeatother

\newcommand{\vspaceline}{\vspace{\baselineskip}}
\newcommand{\newpara}{\vspaceline\noindent}
\newcommand{\newparahead}[1]{\vspaceline\noindent\textbf{#1.}}

\title{Fast-excursion limit of the Heston model}
\author{Ryan McCrickerd\,\orcidlink{0009-0006-8496-7594}\\
Barclays and Imperial College London\\
\href{mailto:ryan.mccrickerd@gmail.com}{ryan.mccrickerd@gmail.com}}
\date{}

\usepackage{fancyhdr}
\begin{document}
\color{bodytext}

\maketitle

\begin{abstract}
This article introduces an unconventional model for price processes in finance that emerges from the classical Heston model under Mechkov's `fast-reversion limit'.
This new \emph{fast-excursion} Heston model exhibits instantaneous (i.e.\ fast) `excursions' through an interval of prices at each time, which are invisible to vanilla options but critical for hitting probabilities and continuously monitored exotics.
Theoretically, the model provides a rare example of a non-degenerate limit of stochastic volatility models that \emph{escapes} the Skorokhod topologies.
This leads to a class of interval-valued processes which exist as lifts of subordinated L\'evy processes, through the concept of selections in the theory of random closed sets.
On the practical side, we show how the model can be simulated using price-time parametric representations, and utilise a purpose-built \emph{classical} Heston simulation scheme in order to visualise convergence.
Finally we demonstrate how this model raises hitting probabilities for barrier options considerably (of order 10\% for one-month EURUSD options), due to taking `excursion risk' into account.

\vspaceline
\end{abstract}

\hspace{5mm}\textbf{Keywords:} fast-reversion limit; random closed sets; barrier options; excursion risk

\hspace{5mm}\textbf{Code repository:} \url{https://github.com/rmcrkd/fast-excursion-limit}

\clearpage

\thispagestyle{plain}
\begin{figure}
\centering
    \includegraphics[width=1.0\textwidth]{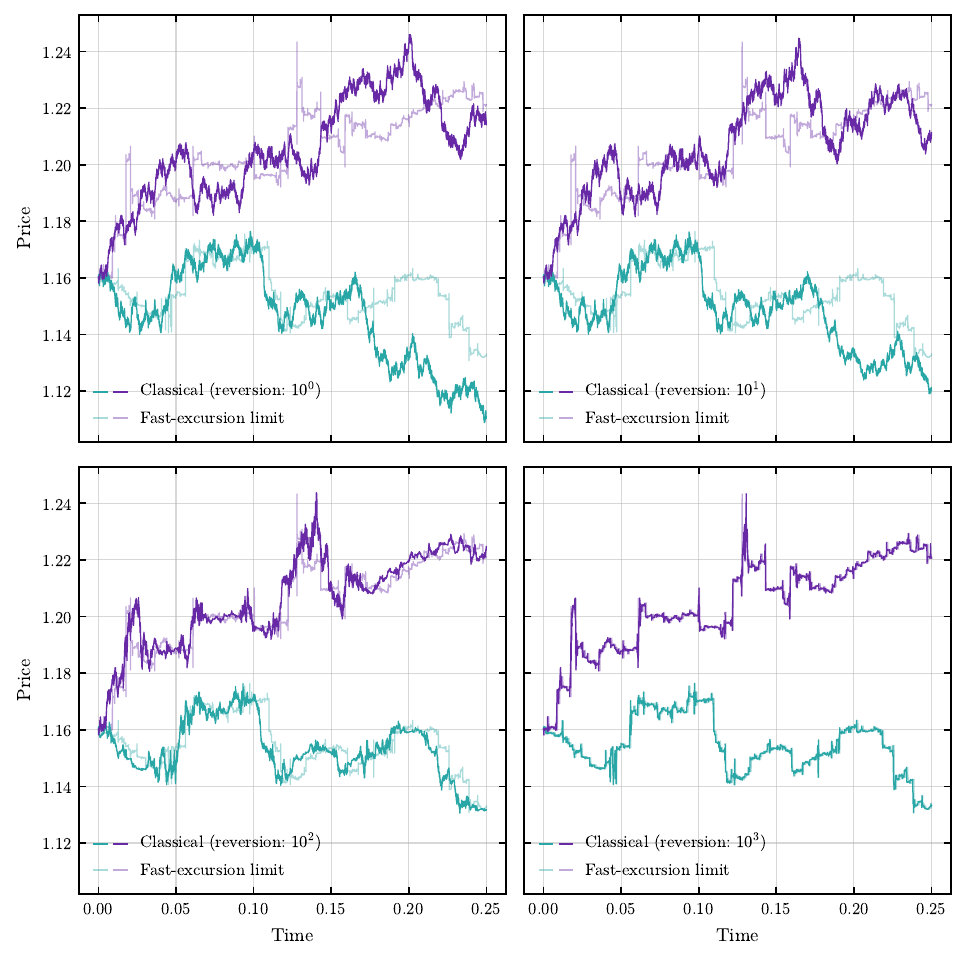}
    \caption{Two classical Heston price paths (dark) converging to their interval-valued\\ `fast-excursion limit' (light).
    Generated by \href{https://github.com/rmcrkd/fast-excursion-limit/blob/main/notebooks/figures/figure-1.ipynb}{figure-1.ipynb}.}
    \label{fig:pathwise-convergence}
\end{figure}

\clearpage
\setcounter{tocdepth}{2}
\tableofcontents
\thispagestyle{plain}

\clearpage
\section{Introduction}\label{sec:introduction}

The classical Heston model \citep{heston_1993} needs no introduction in the contexts of stochastic volatility modelling and derivative pricing.
Moreover it is widely understood that large volatility reversion speeds are required to reconcile certain observations using this model or similar.
Over the past decade, Mechkov's `fast-reversion limit' \citep{mechkov_2015} stands out as a significant contribution towards our understanding and utilisation of models with arbitrarily high reversion speeds.\footnote{
    We allow ourselves to sidestep the `rough volatility' rabbit hole, a class of models also developed over this time, which we could say depend implicitly on \emph{an unbounded continuum of reversionary speeds} \citep{muravlev_2011}.
}
On the practical side, this limiting model has been commercially available over this time (\url{www.numerix.com}), and has recently informed a (publicly available) Heston simulation scheme with outstanding stability \citep{abijaber_2025}.

On the theoretical side, establishing a satisfactory `origin' of this limiting result became one of the main goals of the thesis \citet{mccrickerd_2021}, as it became clear that there is more to it than meets the eye, and that there will be some tangible consequences for derivative pricing.
The main goal of this article is to explain what is going on here, with filtering of and refinements to \citet{mccrickerd_2021} along the way.
Concretely, we now consider \autoref{thm:running-max-limit} to provide the \emph{origin} of Mechkov's result (a \emph{satisfactory} one since this both facilitates visualisations like \autoref{fig:pathwise-convergence} \emph{and} is at the heart of most proofs), and \autoref{thm:fast-excursion-limit} to be the main \emph{consequence} for the classical Heston price process.

\newpara Here is a quick overview of things.
In \citet{mechkov_2015}, Mechkov sends both reversion speed and volatility of volatility (`vol-of-vol') in the Heston model up to infinity at the same rate.\footnote{
    This instigates a violent tussle for the model's instantaneous variance process, while its integral over time proceeds calmly to its (discontinuous) limit, see e.g.\ \citet{mccrickerd_2021} Figure 12.
A similar tussle can be set up for OU processes.
}
And it is shown that, under this limit, European option prices converge to those from a normal-inverse Gaussian (NIG) L\'evy process, which we will henceforth refer to as a fast-reversion Heston (FRH) process or model.
\citet{mccrickerd_2021} extended this marginal result to one between processes, e.g.\ strengthening it to handle (continuously-monitored) path-dependent derivatives.\footnote{
    Separately, \citet{abijaber_2024} extended it to finite-dimensional distributions, noting the wider delicacy in Remark 3.13.
}
In particular, it was shown that the Heston process \emph{does not} converge to an FRH process in any standard topological sense.

This is because, in this fast-reversion limit, the Heston model develops `instantaneous excursions', invisible to European options but meaningful for some exotics.
The emergence of these excursions leads us instead to an \emph{interval-valued generalisation} of the FRH model (with identical parameters), which is a bona fide `excursion process' in the sense of \citet{whitt_2002}, and which we call here the \emph{fast-excursion} Heston (FEH) model.
In the theory of random closed sets \citep{molchanov_2017}, the FRH process is but one `selection process' of the FEH process, with the latter existing on a `hypertopology'.

\newpara The FEH model $\Sigma$ for a price process may be unconventional but it is relatively simple to define; simpler than classical Heston even.
E.g.\ we can construct one by applying elementary operations to a standard 2d Brownian motion $W=(W^0, W^1)$.
As in most stochastic volatility models, we first define a correlated Brownian motion $W^\rho := \sqrt{1 - \rho^2}W^0 + \rho W^1$, then for each time $t\in[0,\infty)$ we can set
\begin{equation}\label{eq:feh-intro}
\Sigma_t = \left\{
    \exp\left(\sigma W^\rho_{x} - \frac12 \sigma^2 x\right) : x\in\left[X_{t-}, X_t\right]
    \right\},\quad X_t = \inf\bigg\{x > 0 : x - \gamma W^1_x > t\bigg\}.\footnote{For brevity, we have fixed a few parameters $s=1, r=q=0$ here as compared with the full \autoref{eq:feh-excursion}.}
\end{equation}
Here $X$ is an inverse Gaussian L\'evy process \citep{applebaum_2009}, and each $\Sigma_t$ can be seen to return a non-trivial \emph{interval of prices} whenever $X$ jumps, i.e.\ whenever $X_t\neq X_{t-}$.
Because such jumps are `unlikely', we find that each $\Sigma_t$ \emph{almost surely contains the single value $S_t=\exp(\sigma W^\rho_{X_t} - \frac12 \sigma^2 X_t )$ only}, which turns out to be that of an FRH process.
This `invisibility' of jumps and excursions for each \emph{fixed} time explains why the FRH and FEH models agree on \emph{European} option prices, and why such prices from the Heston model \emph{do} converge to those from the FRH model, as Mechkov showed.

Along each path, however, jumps and excursions \emph{are not} invisible (despite these occurring almost nowhere).
As a result, the FRH and FEH models \emph{do not} agree on hitting probabilities and certain path-dependent derivative prices like those of touch options.
And it is those prices from the FEH model, exhibiting elevated volatility relative to FRH, to which the Heston model's converge.

\newpara \autoref{sec:fast-excursion-heston} introduces the interval-valued FEH model $\Sigma$ properly, and shows via \autoref{thm:fast-excursion-limit} how it emerges from the classical Heston model under Mechkov's fast-reversion limit.
It also defines a `feature vector' of four real-valued `OHLC' processes arising as projections of $\Sigma$, which help to explain how the model simplifies at parameter boundaries.

Because the finite-dimensional distributions of the FRH and FEH models coincide, vanilla option pricing for the FEH model has been essentially treated already by \citet{mechkov_2015}.
And so in \autoref{sec:touch-pricing} we mainly focus on how single-touch and no-touch option pricing differs between the two models, being arguably the simplest pricing task subject to excursion risk, before concluding.

\autoref{app:background-models} provides an overview of background models like Mechkov's parameterisation of the Heston model and its FRH limit.
Unlike the FEH model, these are all real-valued and c\`adl\`ag, mostly continuous.
Although these details are largely contextual and separate to what follows, viewing the Heston model through the lens of random ODEs as in \autoref{eq:heston-random-ode} is fundamental to both intuition (e.g.\ enabling \autoref{fig:pathwise-convergence}) and proofs relating to convergence.

\newpara This article draws upon theory and limits established in \citet{mccrickerd_2021}, the presentation of which was heavily influenced by Whitt's excursion and fluctuation topologies $\rr{E}$ and $\rr{F}$ \citep{whitt_2002}.
\autoref{thm:fast-excursion-limit} here has since benefited from the theory of random sets in the sense of \citet{molchanov_2017}, on a new space $\rr{G}$ of `graphs'.
An introduction to random (closed) sets is provided in \autoref{app:random-closed-sets}, with derivative pricing in mind.
Things from \autoref{sec:ohlc-processes} onwards do not feature in \citet{mccrickerd_2021} at all.

\section{Fast-excursion Heston}\label{sec:fast-excursion-heston}

This section is in four parts: 1.\ introduces the interval-valued FEH model $\Sigma$ properly in \autoref{def:fast-excursion-heston}, and clarifies how this generalises its FRH namesake through the concept of `selections'; 2.\ shows, via \autoref{thm:fast-excursion-limit}, how this model emerges from the classical Heston model under Mechkov's fast-reversion limit; 3.\ defines a feature vector of four real-valued OHLC processes that arise as projections of the FEH model; and 4.\ uses these OHLC processes to explain simplifications of the model at parameter boundaries.

The notation used for the most important objects in this section is collected in \autoref{tab:notation}.

\begin{table}[ht]
\begin{center}
\begin{tabular}{ll}
 Symbol(s) & Description \\ [0.5ex]
 \hline
 $\bb{G}$ & Price-time domain $\bb{R}\times[0,\infty)$ \\
 $\rr{G}$ & Family of closed subsets of $\bb{G}$; `graphs' \\
 $\cc{G}$ & Effros $\sigma$-algebra on $\rr{G}$ \\
 $W=(W^0,W^1)$ & 2d Brownian motion \\
 $X$ & Inverse Gaussian L\'evy process \\
 $(Z,Y)$ & FEH price-time parametric process \\
 $\Sigma$ & FEH model \\
 $O,H,L,C$ & FEH open, high, low, close processes
\end{tabular}
\end{center}
\caption{Main objects introduced in \autoref{sec:fast-excursion-heston}.}
\label{tab:notation}
\end{table}

\subsection{Model definition and selections}\label{sec:model-definition}

Our FEH model $\Sigma$ will be constructed from the same ingredients as the FRH model from \citet{mechkov_2015}, summarised by \autoref{eq:fast-reversion-heston}.
So for completeness we work from a probability space $(\Omega,\cc{F},\bb{P})$ that supports a 2d Brownian motion $W=(W^0,W^1)$, and have the following fixed \emph{spot, rates, volatility, correlation} and \emph{vol-of-vol} parameters
\[
    s\in[0,\infty),\quad r,q\in\bb{R},\quad \sigma\in[0,\infty),\quad \rho\in[-1,1],\quad \gamma\in[0,\infty).
\]

While we have thus far described our model as an \emph{interval-valued process} (because it is easiest to build intuition this way), this will not be immediately clear from \autoref{def:fast-excursion-heston}.
We should rather think of its trajectories (more simply and generally) as \emph{sets of price-time points} in $\bb{G}:=\bb{R}\times[0,\infty)$, or graphs, and the interpretation as a process over $[0,\infty)$ will follow.
Relevant examples covered by such graphs are the (finite) collection of points used to make \autoref{fig:pathwise-convergence}, or the `completed graphs' of c\`adl\`ag paths \citep{whitt_2002}.
They are formalised by the state space $(\rr{G},\cc{G})$, detailed in \autoref{app:random-closed-sets}.
In summary: $\rr{G}$ is the family of closed subsets of our `carrier space' $\bb{G}$, and $\cc{G}$ is its `Effros' $\sigma$-algebra.

\newparahead{The model} We first introduce a continuous real parametric process $(Z,Y)=\{(Z_x, Y_x)\}_{x\in[0,\infty)}$, representing a price $Z$ at time $Y$, defined by
\begin{equation}\label{eq:feh-fluctuation}
     Z_x := s\exp\bigg((r - q) Y_x +  \sigma W^\rho_x - \frac12 \sigma^2x\bigg),\quad
     Y_x := \max\bigg\{x' - \gamma W^1_{x'}: x'\in[0, x]\bigg\}.
\end{equation}

This `price-time' parametric process $(Z, Y)$ will come to the fore several times in this article, for both theoretical and practical reasons.
For now notice how similar $Z$ is to the Black-Scholes model, summarised by \autoref{eq:black-scholes}, indeed the two coincide on the zero vol-of-vol boundary $\gamma=0$, since then $Y_x=x$.
This parametric process is central to our model definition as follows.

\begin{definition}[Fast-excursion Heston]\label{def:fast-excursion-heston}
    The fast-excursion Heston model is the random closed set $\Sigma:(\Omega,\cc{F},\bb{P})\to(\rr{G},\cc{G})$ traced out by the price-time process $(Z,Y)$, i.e.
    \begin{equation*}
        \Sigma := \big\{(Z_x, Y_x): x\in[0,\infty) \big\}.
    \end{equation*}
\end{definition}

Here \emph{random closed set} is used in the sense of \citet{molchanov_2017} Definition 1.1.1''.
The well-definedness of $\Sigma$ as such is deferred in order to first clarify how this model subtly generalises the FRH model.
To this end, we first project $\Sigma$ onto a compact interval $[\tau, T]$ of times, to obtain a \emph{random compact interval} of prices $\Sigma_{\tau,T}$, which may be written as
\begin{equation}\label{eq:feh-interval}
     \Sigma_{\tau,T} := \bigg\{Z_x : x\in\left[X_{\tau-},X_T\right] \bigg\}, \quad X_t := \inf\bigg\{x > 0 : x - \gamma W^1_x > t\bigg\}.
\end{equation}

Here $X$ is the IG L\'evy process mentioned in \autoref{sec:introduction}, and as usual $ X_{\tau-}:=\lim_{t\uparrow \tau} X_t$ and $ X_{0-}:= X_0=0$.
This representation follows from $[X_{\tau-},  X_T]$ being precisely the set of $x$ for which $Y_x\in[\tau,T]$, i.e.: $X$ and $Y$ are right inverses of one another, e.g.\ satisfying $X_{Y_x}=x$ and $Y_{X_t}=t=Y_{X_{t-}}$.
The case of $ \Sigma_t =  \Sigma_{t,t}$ is similarly a random compact interval, given in full by
\begin{equation}\label{eq:feh-excursion}
     \Sigma_t := \left\{
    s \exp\left((r - q)t + \sigma W^\rho_{x} - \frac12 \sigma^2 x\right) : x\in\left[ X_{t-},  X_t\right]\right\},
\end{equation}
where we have used the fact that $Y_x=t$ for all $x\in[X_{t-},X_t]$.

From \autoref{eq:feh-excursion} we understand that $\Sigma$ is equivalently a compact interval-valued process in the sense of \citet{molchanov_2017} Chapter 5, which we write as $\Sigma=\{\Sigma_t\}_{t\in[0,\infty)}$.
Moreover we see that the (real-valued and c\`adl\`ag) fast-\emph{reversion} Heston model $S$ from \autoref{eq:fast-reversion-heston} is a \emph{selection process} of $\Sigma$, i.e.\ $S_t\in\Sigma_t$ a.s., since each $S_t$ is located in $\Sigma_t$ by selecting $x=X_t$ in \autoref{eq:feh-excursion} specifically.\footnote{Selecting other $x\in[ X_{t-},  X_t]$ provides other selection processes.
Three others will be introduced in \autoref{sec:ohlc-processes}.}

The generalisation (a lift) from $S$ to $\Sigma$ is subtle, however.
On one hand, $X$ being strictly increasing ensures $\Sigma_{\tau,T}$ from \autoref{eq:feh-interval} provides a rich interval of prices for any $T>\tau$, influenced by the excursions evident in \autoref{fig:pathwise-convergence}.
But on the other, since $X_t=X_{t-}$ a.s., each $\Sigma_t$ from \autoref{eq:feh-excursion} considered in isolation is actually indistinguishable from the random singleton $\{S_t\}$.

\newparahead{Well-definedness} We need to show that the FEH model $\Sigma$ in \autoref{def:fast-excursion-heston} is a well-defined random closed set.
Under the additional assumption that $(\Omega,\cc{F},\bb{P})$ is complete, then the fundamental measurability theorem applies in full to $\Sigma$, and equivalently to the set-valued process $\{\Sigma_t\}_{t\in[0,\infty)}$.\footnote{
    This explains \emph{why we define} $\Sigma$ as a random closed set; to minimise room for ambiguity.
E.g.\ if we start from a general (closed) set-valued process defined like in \autoref{eq:feh-excursion}, the resulting $\Sigma$ need not be closed in $\bb{G}$, let alone suitably measurable.
So we end up in the same place, wanting $\Sigma$ to be a random closed set, with all the benefits.
}
See \citet{molchanov_2017} Theorem 1.3.3 and Chapter 5.

For this we need 1.\ $\Sigma\in\rr{G}$ a.s., i.e.\ for $\Sigma$ to be a \emph{closed subset} of $\bb{G}$ a.s., and 2.\ measurability of the map $\Sigma:(\Omega,\cc{F})\to(\rr{G},\cc{G})$, i.e., for all $A\in\cc{G}$,
\[
    \Sigma^{-1}(A):= \{\omega\in\Omega:\Sigma(\omega)\in A\} \in \cc{F}.\footnote{
    Technically we also want the carrier space $\bb{G}$ to be `locally compact Hausdorff second countable' (LCHS), thus Polish. $\bb{G}$ is indeed \citep{molchanov_2017}.
Being LCHS simplifies matters, in particular leading to the equivalence of the Effros $\sigma$-algebra $\cc{G}$ with the Borel $\sigma$-algebra generated by the Fell `hit-and-miss' topology, as discussed in \autoref{app:random-closed-sets}.
}
\]

$\Sigma$ being closed in $\bb{G}$ follows from the continuity of $(Z,Y)$ along with $Y$ being non-decreasing, with $Y_0=0$ and $Y_x\to\infty$ a.s.\ as $x\to\infty$.
Regarding measurability we can do better.
We can consider $\Sigma$ as a composition of maps $\Sigma=m_2\circ m_1\circ m_0$ starting from $(\Omega,\cc{F})$ with
\[
    m_0 = W,\quad m_1: W\mapsto(Z,Y),\quad m_2:(Z,Y)\mapsto\Sigma.
\]

We can ignore $m_0$ by w.l.o.g.\ assuming $(\Omega,\cc{F},\bb{P})$ is the canonical space supporting $W$.
It is straightforward to see that $m_1$ is continuous, and the continuity of $m_2$ follows from \autoref{lem:attouch-wets}, as utilised in \autoref{thm:fast-excursion-limit}.
As a composition of Borel measurable functions, $\Sigma$ is measurable.

Having $\Sigma$ as a well-defined random closed set ensures a well-defined distribution $\mu_{\Sigma} := \bb{P}\Sigma^{-1}$ on $(\rr{G},\cc{G})$.
Other statements e.g.\ regarding random compact intervals and singletons follow similarly.

\newparahead{Discussion} An interval-valued process defined like in \autoref{eq:feh-excursion} first appeared in the thesis \citet{mccrickerd_2021}.
The presentation there was heavily inspired by Whitt's excursion topology $\rr{E}$, and related fluctuation topology $\rr{F}$.
See \citet{whitt_2002} Chapter 15.
Indeed the $\rr{E}$ in FEH pays homage, and our parametric representation $(Z,Y)$ lives in $\rr{F}$ (actually the equivalence class of such processes which induce $\Sigma$).\footnote{
    See \citet{abijaber_2026} and references therein for some history and developments of these spaces.
In particular the so-called `decorated c\`adl\`ag paths' $\mathfrak{D}$ and underlying space $D^\circ$ of parametric representations are akin to Whitt's $\rr{E}$ and $\rr{F}$.
}

We now consider the excursion space $\rr{E}$ (as Whitt later defines it in terms of continuous parametric representations) to contain \emph{special realisations} within the wider (and arguably simpler) space $\rr{G}$.
See the section in \citet{molchanov_2017} on \emph{random closed sets with special realisations}.
It is desirable to define our model in $\rr{G}$ like this, with a distribution $\mu_\Sigma$ on the Effros $\sigma$-algebra $\cc{G}$, because this space is \emph{tremendously} robust.
E.g.\ the Fell topology on $\rr{G}$, which generates $\cc{G}$, is compact and metrisable, see e.g.\ \citet{beer_1993} Theorem 5.1.5.\footnote{
    It is the inclusion of $\varnothing$ in $\rr{G}$ which compactifies this Fell topology, like $\infty$ for the one-point compactification of $\bb{R}$.
In practice we can work with the subset $\rr{G}':=\rr{G}\setminus\varnothing$, the Fell topology of which is `just' Polish, and is specifically induced by the \emph{Attouch-Wets} metric, as used in \autoref{thm:fast-excursion-limit}.
This metric broadly reflects `Hausdorff convergence over compacts' (while elegantly navigating issues at compacts' boundaries).
See e.g.\ \citet{beer_1993} Chapter 3.}

As a bonus, the fundamental hit-and-miss measurable events in $(\rr{G},\cc{G})$ are strikingly relevant to derivative pricing through random closed sets' capacity functionals.
And Choquet's theorem essentially provides classes of `simple' derivatives from which any such distribution $\mu_\Sigma$ on $(\rr{G},\cc{G})$ can be uniquely reconstructed.
I.e.\ it provides (theoretical) classes of (actual) derivatives, quotes for which would make our models completely \emph{model-free}.
See \autoref{app:random-closed-sets} or \citep{molchanov_2017} more generally.

\newpara We have not mentioned filtrations above and they will not be important henceforth either, but notice that $\Sigma_t$ in \autoref{eq:feh-excursion} is adapted to the \emph{time-changed} filtration $\hat{\cc{F}}_t:=\cc{F}_{X_t}$ where $\cc{F}_x:=\sigma(\{W_{x'}\}_{x'\in[0,x]})$.
Where we need to tell such objects apart from same-named counterparts, we mark them with hats $\hat{}$, these being precisely the ones which are not hampered by SDEs and for which \emph{a.s.} limiting results will be possible.
See e.g.\ the use of SDE solution $S$ and random ODE solution $\hat S$ in the proof of \autoref{thm:fast-excursion-limit}.
Such $\hat{}$ processes yield limit theorems that embody Skorokhod's representation theorem (enabling visualisation like in \autoref{fig:pathwise-convergence}), as opposed to Prokhorov's theorem \citep{billingsley_1999}.

\subsection{Convergence under the fast-reversion limit}\label{sec:convergence}

This section provides \autoref{thm:fast-excursion-limit}, which regards the weak convergence of Heston model graphs under Mechkov's fast-reversion limit.
By weak convergence $\Sigma^a\implies\hat \Sigma$ we mean in the sense of \citet{billingsley_1999}.\footnote{
    See e.g.\ the Portmanteau Theorem for a definition and alternative characterisations.
For concreteness, following \citet{billingsley_1999} Section 3, we can take weak convergence $\Sigma^a\implies\hat \Sigma$ to mean $\bb{E}[f(\Sigma^a)]\to\bb{E}[f(\hat \Sigma)]$ for all bounded $f:\rr{G}\to\bb{R}$ which are continuous from the Fell topology on $\rr{G}$.
While this is clearly helpful for derivative pricing (when conducted via expectations), continuity should not be taken for granted for a general derivative payoff.
More on this in \autoref{sec:touch-pricing}.
}
By Heston model graphs we mean the random closed sets $\Sigma^a$ induced by the Heston model $S^a$ as in \autoref{eq:mechkov-heston}.
For convenience, we therefore have the models
\begin{equation}\label{eq:thm-heston}
    \frac{dS^a_t}{S^a_t} = (r - q)dt + \sigma\sqrt{V^a_t} dW^\rho_t,\quad dV^a_t = a \left((1 - V^a_t)dt + \gamma\sqrt{V^a_t}dW^1_t\right),
\end{equation}
where we fix each $(S^a_0,V^a_0)=(s,v)\in(0,\infty)^2$, and corresponding graphs
\begin{equation}\label{eq:thm-graphs}
\Sigma^a:=\{(S^a_t,t):t\in[0,\infty)\}\quad \iff\quad \Sigma^a=\{\Sigma^a_t\}_{t\in[0,\infty)},\quad \Sigma^a_t := \{S^a_t\}.
\end{equation}

Convergence of these graphs $\Sigma^a$ as $a\to\infty$ will be established towards our FEH model $\hat \Sigma$ from \autoref{def:fast-excursion-heston}, which can be summarised as in \autoref{eq:feh-excursion} by its marginal random intervals, where $\hat X$ is the IG L\'evy process of \autoref{eq:feh-interval},
\begin{equation}\label{eq:thm-limit}
    \hat \Sigma_t = \left\{
    s \exp\left((r - q)t + \sigma W^\rho_{x} - \frac12 \sigma^2 x\right) : x\in\left[ \hat X_{t-}, \hat X_t\right]
    \right\}.
\end{equation}

\newparahead{Statement} The formal statement of the result then goes as follows, wherein we recall that $\rr{G}$ is the family of closed subsets of our price-time domain $\bb{G}$, containing our models $\Sigma^a$ and $\hat \Sigma$ a.s., and that $\cc{G}$ is the Effros $\sigma$-algebra on $\rr{G}$, on which the distributions of these models are defined.

\begin{theorem}[Fast-excursion limit]\label{thm:fast-excursion-limit}
    For each $a>0$, let $S^a=\{S^a_t\}_{t\in[0,\infty)}$ be the Heston model as in \autoref{eq:thm-heston},
    and let $\Sigma^a=\{\Sigma^a_t\}_{t\in[0,\infty)}$ be its graph as in \autoref{eq:thm-graphs}, i.e.\ its canonical embedding into $\rr{G}$.
    Let $\hat \Sigma=\{\hat \Sigma_t\}_{t\in[0,\infty)}$ be the fast-excursion Heston model as in \autoref{eq:thm-limit}.
    Then, considering each $\Sigma^a$ and $\hat \Sigma$ as random elements of the measurable space $(\rr{G},\cc{G})$, the weak convergence $\Sigma^a\implies\hat \Sigma$ takes place under Mechkov's fast-reversion limit, i.e.\ as $a\to\infty$.
\end{theorem}

\noindent\textbf{Proof.} At the heart of this is the `Heston random ODE' from \autoref{eq:heston-random-ode}, and the related convergence result \autoref{thm:running-max-limit}, specialised from \citet{mccrickerd_2021}.
In three steps:

\vspaceline\noindent\textbf{Step 1} (a.s.\ uniform convergence of parametric representations)\textbf{.} Let $\hat S^a$ denote Heston random ODE processes as in \autoref{eq:heston-random-ode}, depending on strictly increasing cumulative variance processes $\hat X^a$, with inverses $\hat Y^a:=(\hat X^a)^{-1}$.
Then the associated graphs $\hat \Sigma^a$ of $\hat S^a$ admit the representations
\begin{equation*}
\hat \Sigma^a:= \{(\hat S^a_t,t):t\in[0,\infty)\} = \{(\hat Z^a_x,\hat Y^a_x):x\in[0,\infty)\},
\end{equation*}
where
\[
\hat Z^a_x:=\hat S^a_{\hat Y^a_x} = s \exp\left((r - q)\hat Y^a_x + \sigma W^\rho_{x} - \frac12 \sigma^2 x\right).
\]
Let $(\hat Z,\hat Y)$ denote the FEH parametric representation from \autoref{eq:feh-fluctuation}.
Now utilising \autoref{thm:running-max-limit}, we obtain $\hat Y^a\to\hat  Y$ a.s.\ uniformly over compacts as $a\to\infty$, and therefore $\hat Z^a\to \hat Z$ similarly.

\vspaceline\noindent\textbf{Step 2} (a.s.\ convergence on the Fell topology)\textbf{.} Having $(\hat Z^a,\hat Y^a)\to(\hat Z,\hat Y)$ a.s.\ uniformly over compacts, use \autoref{lem:attouch-wets} to obtain $\hat \Sigma^a\to\hat \Sigma$ a.s.\ w.r.t.\ the Attouch-Wets distance.
The Attouch-Wets distance is a bona fide metric on the subset $\rr{G}':=\rr{G}\setminus\{\varnothing\}$, inducing its Fell topology.
So we obtain $\hat \Sigma^a\to\hat \Sigma$ a.s.\ on the Fell topology of $\rr{G}'$, and also on the Fell topology of $\rr{G}=\rr{G}'\cup\{\varnothing\}$.

\vspaceline\noindent\textbf{Step 3} (weak convergence on the Effros $\sigma$-algebra)\textbf{.} The Fell topology on $\rr{G}$ generates its Effros $\sigma$-algebra $\cc{G}$, so we obtain the weak convergence $\hat \Sigma^a\implies\hat \Sigma$ on $(\rr{G},\cc{G})$.
Since the processes $\hat S^a$ and $S^a$ share the same distribution, so too do the induced random closed sets $\hat\Sigma^a$ and $\Sigma^a$ on $(\rr{G},\cc{G})$.
So we obtain the weak convergence $\Sigma^a\implies\hat\Sigma$ on $(\rr{G},\cc{G})$ as claimed, completing the proof.\qed

\newparahead{Discussion} \autoref{thm:fast-excursion-limit} differs only slightly from \citet{mccrickerd_2021} Corollary 4.58, which essentially regards weak convergence of the same processes to the same limit, but on a more robust space (as discussed in \autoref{sec:fast-excursion-heston}).
As such it is only step 2 above which changes meaningfully, with the Attouch-Wets metric replacing the notion of Hausdorff convergence over compacts on Whitt's excursion space $\rr{E}$.

Much clarity on these spaces and this result is brought through \autoref{fig:pathwise-convergence}, although notice that what we are visualising here is really the intermediate \emph{a.s.\ }result $\hat \Sigma^a\to\hat \Sigma$ w.r.t.\ the Attouch-Wets distance.
The corresponding weak result $\Sigma^a\implies\hat\Sigma$ cannot be visualised for a single trajectory like this, and is best understood through consequences for derivative prices as in \autoref{sec:touch-pricing}.

\newpara It is natural to ask if \autoref{thm:fast-excursion-limit} holds when processes are restricted to a compact time domain $[0,T]$.
The answer is yes, e.g.\ leading to the convergence of derivative prices over such time horizons as in \autoref{sec:touch-pricing}, but this is \emph{specifically} because of the stochastic continuity of $\hat X$, i.e.\ $\hat X_{T-}=\hat X_T$ a.s.
In general, the Fell topology on $\rr{G}$, the Hausdorff topology on $\rr{E}$ and Skorokhod's M$_2$ topology on $\rr{D}$ are identical in this regard: restrictions and maps thereof like $f\mapsto \sup_{t\in[0,T]}f(t)$ from $\rr{D}$ to $\bb{R}$ are \emph{not} continuous, essentially because time and space are placed on equal footing in these topologies, and we might have a discontinuity or excursion at the end time $T$.
See \citet{billingsley_1999} and \citet{whitt_2002}.

\newpara From a practical perspective notice that in the passage from Heston model $S$ to FEH model $\hat\Sigma$ in \autoref{thm:fast-excursion-limit}, we have lost \emph{two} parameters; $a$ and $v$.
This leaves us with the gold standard \emph{three} volatility-related parameters $\sigma$, $\rho$ and $\gamma$, which preserve their interpretation and effect on volatility surfaces from the Heston model, and can be made time-dependent under extensions of the FEH model.
All of these points were raised first in \citet{mechkov_2015}; they apply equally to the FRH model.

\subsection{OHLC processes}\label{sec:ohlc-processes}

In this section four real-valued processes (with left and right limits) are introduced, which arise as certain projections of the FEH model $\Sigma$.
Each of these is a selection process of $\Sigma$, like the FRH model is, as discussed in \autoref{sec:model-definition}.
Collectively, they provide a `feature vector' of $\Sigma$ in the sense of \citet{hess_1999}.\footnote{
    To provide an example, we can write $\Sigma_t$ as an interval $[L_t, H_t]$.
While $\Sigma_t$ lives on a hyperspace of $\bb{R}$ (it is a random element of the \emph{closed subsets of} $\bb{R}$), each of $L_t$ and $H_t$ rather live on the carrier space $\bb{R}$.
Being a finite vector, we call $(L_t,H_t)$ a feature vector of $\Sigma_t$.
}

It is expected that just one (or max two) of these are sufficient for most applications, like derivative pricing under $\Sigma$ (see \autoref{sec:touch-pricing}).
But mainly they provide convenient means to derive and state certain properties of $\Sigma$, especially how $\Sigma$ is related to and generalises more familiar (real-valued) processes, and how things change under certain parameter regimes (as in \autoref{sec:parameters}).

\newparahead{Definitions} Considering the random interval $\Sigma_{\tau,T}$ of prices from \autoref{eq:feh-interval}, four \emph{open, high, low} and \emph{close} (OHLC) random prices arise from the extrema of the relevant price and time intervals,
\begin{equation}\label{eq:ohlc}
    O_\tau:= Z_{X_{\tau-}},\quad H_{\tau,T}:= \max \Sigma_{\tau,T},\quad L_{\tau,T}:= \min \Sigma_{\tau,T},\quad C_T:= Z_{X_{T}}.
\end{equation}

Allowing high and low prices for a single time also, i.e.\
\begin{equation*}
    H_t:= \max \Sigma_t,\quad L_t:= \min \Sigma_t,
\end{equation*}
leads us to four OHLC \emph{processes}, $O, H, L, C$, e.g.\ $C=\{C_t\}_{t\in[0,\infty)}$.\footnote{
    For the avoidance of doubt, all of these min/maxes are actually achieved.
}
Notice how the (c\`adl\`ag) FRH process from \autoref{eq:fast-reversion-heston} coincides with the close price.
These are all illustrated in \autoref{fig:candlestick}.

\begin{figure}[ht]
\centering
    \includegraphics[width=0.95\textwidth]{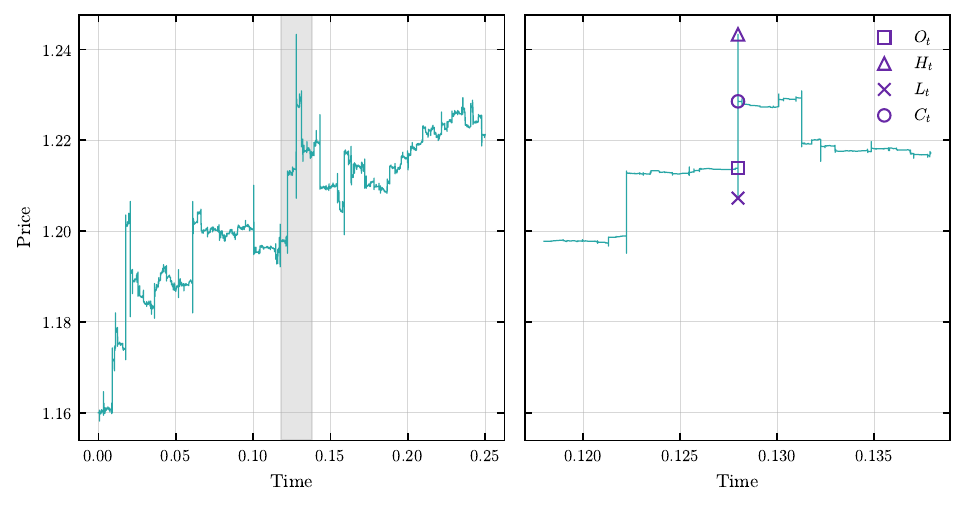}
    \caption{The right panel shows OHLC prices $\{O_t, H_t, L_t, C_t\}\subset \Sigma_t=[L_t,H_t]$.
This right panel is a zoom-in of the shaded region in the left panel.}
    \label{fig:candlestick}
\end{figure}

Like $H$ and $L$, $O$ and $C$ can be defined directly from $\Sigma$ as well, i.e.\ we do not actually need $Z$ or $X$ as used in \autoref{eq:ohlc}.
For this we take Hausdorff limits of $\Sigma$ in time as in \citet{whitt_2002}, and obtain
\begin{equation*}
    \Sigma_t\xrightarrow{t\uparrow \tau} \{O_\tau\},\quad \Sigma_t\xrightarrow{t\downarrow T} \{C_T\},
\end{equation*}
a.s., which we can equivalently write as random singleton left and right limits, e.g.\ $\Sigma_{t+} = \{C_t\}$.\footnote{
    Although these limits do not \emph{follow} from e.g.\ \citet{whitt_2002} Theorem 15.4.1 (these limits are rather presented as conditions for compact graphs), they are straightforward to establish given the continuous parametric representation $(Z,Y)$ of $\Sigma$.
The existence of parametric representations is later assumed in \citet{whitt_2002}, see the discussion following Theorem 15.4.5.
}
This clarifies that the OHLC processes are bona fide (extrema or limiting) projections of $\Sigma$, e.g.\ since $\max\max\Sigma_t=\max\Sigma_t$ (after taking the canonical embedding of $H=\max\Sigma$ back into $\rr{G}$), and $\Sigma_{t++}=\Sigma_{t+}$.
Going the other way, $\Sigma$ can be reconstructed from $H$ and $L$ using $\Sigma_t=[L_t,H_t]$.

\newparahead{Properties} These OHLC processes are fascinating in their own right.
E.g.\ we have $O_t = H_t= L_t = C_t$ a.s.\ (since $X_t=X_{t-}$ a.s.) for each \emph{fixed} $t\in[0,\infty)$, with all having the same exponentiated NIG distribution as the FRH process (avoiding degeneracies at parameter boundaries).
E.g.\ setting $N_t:=\log(\frac{S_t}{s}) - (r-q)t$ for any of $S\in\{O,H,L,C\}$, we have $\mathbb{E}[e^{ip N_t}]=e^{\zeta_N(p)t}$, with characteristic exponent precisely as in \autoref{eq:nig-exponent}, i.e.\
\begin{equation}\label{eq:ohlc-char-exponent}
    \zeta_N(p) = - \frac{\sigma\rho}{\gamma}ip + \frac{1}{\gamma^2}\left(1 - \sqrt{(1 - i\sigma\rho\gamma p)^2 + (\sigma\gamma)^2(p^2 + ip)}\right),
\end{equation}
and $\zeta_N(-i)=0$, i.e.\ $\mathbb{E}[e^{N_t}]=1$, provided $\sigma\rho\gamma\le 1$.
These OHLC processes are moreover equivalent a.e.\ (w.r.t.\ Lebesgue) a.s., have the same left limits e.g.\ $H_{t-}=O_t$ and right limits e.g.\ $L_{t+}=C_t$, same Riemann integrals, same $L^p$ norms, same essential infimums/supremums etc.
Yet they are nevertheless distinct processes, only with the following ordering in general
\begin{equation}\label{eq:ohlc-ordering}
    L \le O\wedge C \le O\vee C \le H.\footnote{
        Concretely, what this means is: for any such $[\tau,T]$, we have $L_t \le O_t\wedge C_t \le O_t\vee C_t \le H_t$ for all $t\in[\tau,T]$ a.s.
    }
\end{equation}

One further regularity property worth pointing out is that from \citet{whitt_2002} Theorem 15.4.1, which bounds the frequency of large excursions: for each $\epsilon>0$, there are a.s.\ only finitely many $t\in[\tau,T]$ for which $|\Sigma_t|:=H_t-L_t > \epsilon$.
(Again this is straightforward to prove using $(Z,Y)$.)
The equivalent statement for $O_t\vee C_t - O_t\wedge C_t >\epsilon$ follows from $C$ being an exponentiated c\`adl\`ag L\'evy process with left limits $C_{t-}=O_t$.

\newparahead{Excursions} We have just alluded to an excursion using a condition like $H_t-L_t > 0$.
But in light of \autoref{fig:candlestick} we should distinguish between \emph{proper} excursions (of the FEH model) and jumps (e.g.\ of its FRH selection).
We should more accurately consider a model like $\Sigma$ to possess \emph{proper} excursions over some $[\tau,T]$ if one of the following conditions holds with non-zero probability,
\begin{equation*}
    \text{upward:}\ \max_{t\in[\tau,T]} (H_t - O_t\vee C_t) > 0,\quad \text{downward:}\ \min_{t\in[\tau,T]}(L_t - O_t\wedge C_t) < 0.
\end{equation*}
While it seems reasonable to conjecture that these inequalities actually hold a.s.\ for \emph{most} parameter sets, we will not pursue this in detail.
Rather, \autoref{sec:parameters} shows some parameter boundaries where proper excursions are \emph{not} present, and \autoref{sec:touch-pricing} presents numerical consequences when they are.

\newparahead{Discounting} In the following section it is helpful to work with \emph{log-discounted} OHLC processes $\log \tilde S_t = \log(\frac{S_t}{s}) - (r-q)t$ for price processes $S\in\{O,H,L,C\}$, mapping conclusions onto $\Sigma$.
The resulting processes $\tilde O,\tilde H,\tilde L,\tilde C$ should be understood as being the OHLC processes for a \emph{discounted} FEH model $\tilde \Sigma$ (e.g.\ $\tilde H_t=\max\tilde\Sigma$) which can be summarised by
\begin{equation}\label{eq:disc-excursion}
    \log \tilde \Sigma_t = \log\left(\frac{\Sigma_t}{s}\right) - (r-q)t = \bigg\{\sigma  W^\rho_x - \frac12 \sigma^2 x: x\in \left[X_{t-}, X_t\right]\bigg\}.\footnote{Here we use set operations like those spelled out in \citet{beer_1993} Chapter 1, e.g.\ $\log\left(\frac{\Sigma_t}{s}\right) = \{\log\left(\frac{a}{s}\right):a\in\Sigma_t\}$.}
\end{equation}
This relies on operator reorderings, e.g.\ $\log\max\Sigma_t=\max\log\Sigma_t$, which we will do without warning.

\subsection{Behaviour at parameter boundaries}\label{sec:parameters}

This section considers behaviour of the FEH model $\Sigma$ at the correlation boundaries $\rho=\pm1$, vol-of-vol boundary $\gamma=0$, and under the limit $\gamma\to\infty$.
Behaviour w.r.t.\ other parameters (spot $s$, rates $r,q$ and volatility $\sigma$) is practically identical to that of the Black-Scholes model.

\newparahead{Correlation} We understand the behaviour of $\Sigma$ w.r.t.\ correlation $\rho$ via Brownian excursions and the OHLC processes, focussing on the discounted process $\tilde\Sigma$ from \autoref{eq:disc-excursion}.
Exposing $\rho$ and using $\bar\rho:=\sqrt{1-\rho^2}$, we have
\begin{equation}\label{eq:log-excursion}
    \log \tilde\Sigma_t = \bigg\{\sigma \bar\rho W^0_x + \sigma\rho W^1_x - \frac12 \sigma^2 x: x\in X^*_t\bigg\},\quad X^*_t:=\left[X_{t-}, X_t\right],
\end{equation}
where $X^*_t$ is only introduced here to compress some expressions later.

By the end we should understand that the term $\sigma\rho W^1_x$ here is fighting for excursions \emph{only} in the direction of $\rho$, because of how $W^1$ and $X$ are related.\footnote{
    This is not really a novelty.
For the FRH model from \autoref{eq:fast-reversion-heston}, the same is true of \emph{jumps} resp.\ excursions.
Both follow from $\max\{x - \gamma W^1_x:x\in X^*_t\}=t$.
See the concluding statements of \citet{mccrickerd_2021} Chapter 3 and Figure 15.
}
So if we kill off the $\sigma \bar\rho W^0_x$ term by moving to the boundary $\rho=-1$, then we find that $\Sigma$ has \emph{downwards} excursions only, as is evident from \autoref{fig:correlation-boundary}.
(Definitively, we have equality in the usual $H \ge O\vee C$ from \autoref{eq:ohlc-ordering}.)
We will focus on this case of $\rho=-1$, which leads to an elegant wider OHLC ordering,
\begin{equation}\label{eq:ohlc-ordering-rho}
    \rho=-1\ \implies\ L \le C \le O = H.
\end{equation}

\begin{figure}[ht]
\centering
    \includegraphics[width=0.95\textwidth]{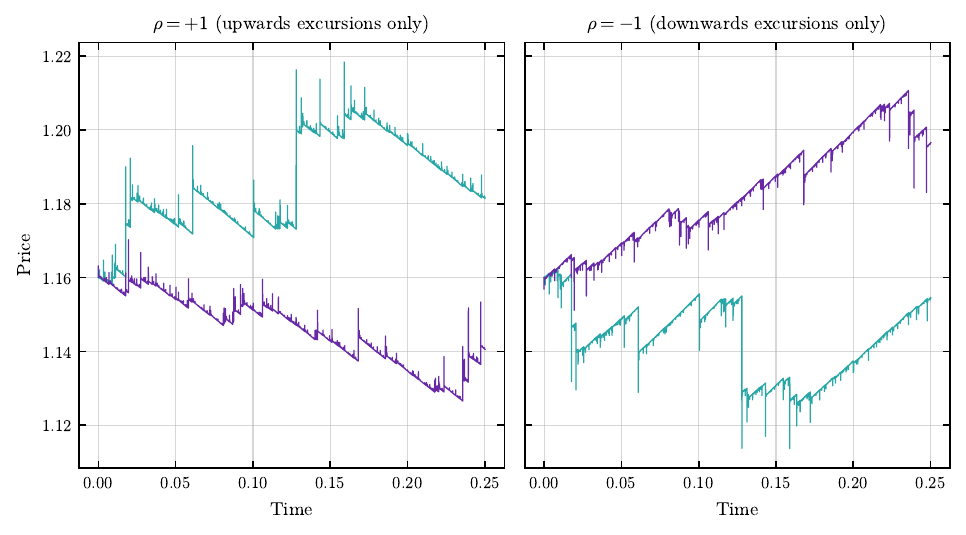}
    \caption{FEH paths on the correlation parameter boundaries $\rho=\pm1$.}
    \label{fig:correlation-boundary}
\end{figure}

To obtain \autoref{eq:ohlc-ordering} we look at excursions of the drifting Brownian $B_x:=x - \gamma W^1_x$ which defines $X$.
With $X$ being the right inverse of $B$, we have $B_x=t$ for both $x=X_{t-}$ and $x=X_t$, but not for all $x\in X^*_t$.
We rather have the following, which are key to bounding excursions of $\tilde\Sigma$,
\begin{equation*}
    \min\big\{B_x:x\in X^*_t\big\}\le t,\quad \max\big\{B_x:x\in X^*_t\big\}=t.
\end{equation*}

First, prioritising $B$ over $W^1$ in \autoref{eq:log-excursion} and setting $\rho=-1$, we get
\begin{equation*}
    \log \tilde \Sigma_t = \bigg\{ \frac{\sigma}{\gamma}B_x - \sigma\left(\frac{\sigma}{2} +  \frac{1}{\gamma}\right)  x : x\in X^*_t\bigg\}.
\end{equation*}

We will deal with the case of $\gamma=0$ (for all $\rho$) shortly, so here assume $\gamma>0$.
Now using the endpoints $x=X_{t-}$ and $x=X_t$ (where $B_x=t$), note the following open/close expressions and ordering
\begin{equation}\label{eq:frh-ordering}
    \log \tilde C_t = \frac{\sigma}{\gamma}t - \sigma\left(\frac{\sigma}{2} +  \frac{1}{\gamma}\right)  X_t \le \frac{\sigma}{\gamma}t - \sigma\left(\frac{\sigma}{2} +  \frac{1}{\gamma}\right)  X_{t-} = \log \tilde O_t.
\end{equation}

Now considering $\log\tilde H_t=\log\max\tilde\Sigma_t=\max\log\tilde\Sigma_t$ and utilising the bound $\max_{x\in X^*_t}B_x=t$, we find $\tilde O$ appearing,
\begin{equation}\label{eq:ordering-analysis}
    \log \tilde H_t \le \frac{\sigma}{\gamma}\max_{x\in X^*_t}B_x - \sigma\left(\frac{\sigma}{2} +  \frac{1}{\gamma}\right)\min_{x\in X^*_t}x = \frac{\sigma}{\gamma}t - \sigma\left(\frac{\sigma}{2} +  \frac{1}{\gamma}\right)X_{t-} = \log \tilde O_t.
\end{equation}

This $\tilde H\le \tilde O$ prohibits positive excursions, and combined with the general $\tilde H\ge \tilde O$ gives $\tilde H= \tilde O$.
Along with the general $\tilde L\le \tilde C$ and $\tilde C\le \tilde O$ from \autoref{eq:frh-ordering}, we get the full ordering in \autoref{eq:ohlc-ordering-rho}.

Separately, similar analysis to that in \autoref{eq:ordering-analysis} gives the following bound for $\tilde L_t$,
\begin{equation*}
    \log \tilde L_t \le \frac{\sigma}{\gamma}\min_{x\in X^*_t}B_x - \sigma\left(\frac{\sigma}{2} +  \frac{1}{\gamma}\right)X_{t-}.
\end{equation*}

The analysis for $\rho=1$ is similar but in place of e.g.\ $(\frac{\sigma}{2} +  \frac{1}{\gamma})\min_{x\in X^*_t}x$ in \autoref{eq:ordering-analysis} we arrive at $\max_{x\in X^*_t}(\frac{\sigma}{2} - \frac{1}{\gamma})x$, which differs depending on the cases of $\frac{1}{\gamma}\gtreqqless\frac{\sigma}{2}$.
The main point holds in that we now have $L = O\wedge C$ (i.e.\ upwards excursion only).
Typical parameters also lead to $O \le C \le H$ as in \autoref{fig:correlation-boundary}, but whether $O \lesseqqgtr C$ holds depends on $\frac{1}{\gamma}\gtreqqless\frac{\sigma}{2}$ in this $\rho=1$ case.

\newparahead{Vol-of-vol} The behaviour of $\Sigma$ w.r.t.\ vol-of-vol $\gamma$ can be understood through the normalised IG L\'evy process $X$ defined in \autoref{eq:feh-interval}.
Recall that $X_t\sim \rr{IG}(t, \frac{t^2}{\gamma^2})$ for $\gamma, t>0$, and by simply setting $\gamma=0$ in \autoref{eq:feh-interval} then $X_t=t$.
So for all $\gamma, t\ge0$ we have $\bb{E}[X_t]=t$ and $\bb{V}[X_t]=\gamma^2 t$.

\begin{figure}[ht]
\centering
    \includegraphics[width=0.95\textwidth]{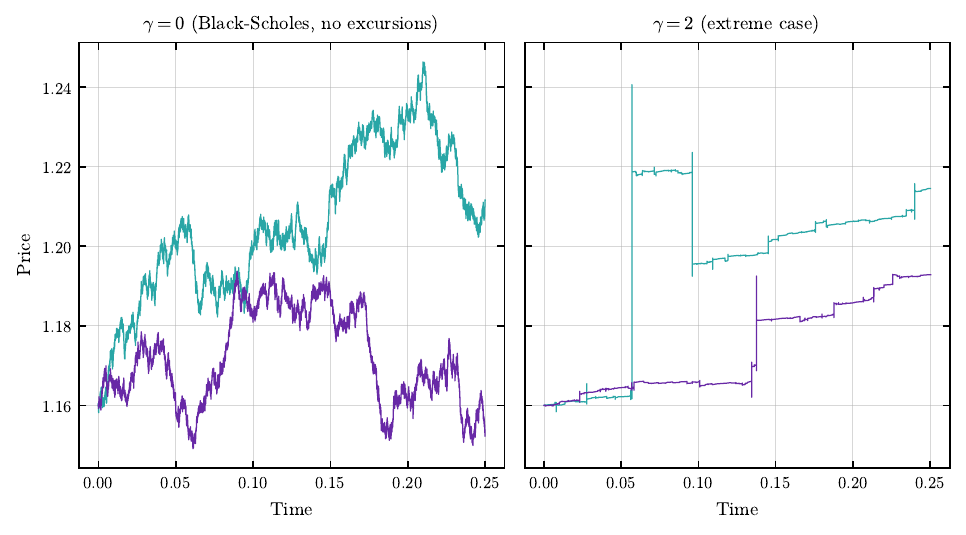}
    \caption{FEH paths on the vol-of-vol parameter boundary $\gamma=0$ and for an extreme case $\gamma=2$.}
    \label{fig:vol-of-vol-boundary}
\end{figure}

Setting $\gamma=0$ and therefore putting $X_t=t$ into \autoref{eq:feh-excursion}, we find
\begin{equation}\label{eq:canonical-bs}
    \gamma=0\ \implies\ \Sigma_t = \left\{
    s \exp\left((r - q)t + \sigma W^\rho_t - \frac12 \sigma^2 t\right) \right\}.
\end{equation}
So the FEH model is now just a random singleton at each time, containing the Black-Scholes price.
Put another way, it \emph{is} the Black-Scholes model after this model's canonical embedding into $(\rr{G},\cc{G})$.
Unexciting but reassuring, and evident in \autoref{fig:vol-of-vol-boundary}.

The actual convergence as $\gamma\to 0$ can be characterised by either establishing convergence of \emph{both} of the high and low processes $H,L$ to Black-Scholes (recall $\Sigma_t=[L_t,H_t]$), or alternatively we can work directly with $\Sigma$, establishing convergence of this to Black-Scholes per \autoref{eq:canonical-bs}.
In either case convergence is a.s.\ and uniform over compacts.

Going the other way $\gamma\to\infty$ is mathematically interesting but (it would appear) less practically relevant.
The behaviour of each $X_t$ is summarised by the convergence in probability to zero, but not in mean (recall $\bb{E}[X_t]=t$ for all $\gamma$) and certainly not in higher-order $L^p$ spaces.
Indeed,
\[
    \forall\epsilon>0\ \ \bb{P}[X_t>\epsilon] \to 0,\ \ \text{but}\ \ \forall p > 1\ \ \bb{E}[X^p_t]\to\infty, \ \ \text{as}\ \ \gamma\to\infty.
\]
Counterpart statements on the \emph{process} $X$ follow from $X_0=0$ and it being strictly increasing.

So $X_t$ starts to spend a lot of time near zero before experiencing a monstrous and practically unrealistic jump, and likewise $\Sigma_t$ stays close to the singleton $\{se^{(r - q)t}\}$ before a large excursion, as is clear from \autoref{fig:vol-of-vol-boundary}.
Equivalently, both $H$ and $L$ stay close to $se^{(r - q)t}$ then suddenly diverge.

It is interesting to note that this degeneracy can be stabilised by working with $\bar X_t:=\gamma^2 X_t$.
This instead converges in distribution to a L\'evy$(0,t^2)$ variable as $\gamma\to\infty$, as in \citet{mccrickerd_2021} Table 1, but this moves further away from practical relevance since now even $\bb{E}[\bar X_t]\to\infty$.

\section{Touch option pricing}\label{sec:touch-pricing}

This section introduces derivative pricing under the FEH model $\Sigma$.
After reviewing how to price vanillas, we focus on the pricing of single-touch and no-touch barrier options, being arguably the simplest which expose and develop intuition for excursion risk.

When providing such a derivative price under the FEH model, we have one clear agenda: \emph{to state the real and unique price to which the classical Heston model's converge under Mechkov's fast-reversion limit}, as a consequence of \autoref{thm:fast-excursion-limit}.\footnote{
    Throughout this section we use $S^a=\{S^a_t\}_{t\in[0,\infty)}$ to denote the Heston model as appearing in \autoref{thm:fast-excursion-limit}.
}
In this way the FEH model becomes a mechanism to help us access a certain boundary of a \emph{classical} (arbitrage-free) pricing framework.\footnote{
    This appears to be the first case of an interval-valued \emph{price} process used to access the boundary of a classical framework.
As opposed to e.g.\ being used to obtain interval-valued \emph{derivative} prices as in conic finance \citep{madan_2016}.
}

We will see that our derivative prices under the FEH model can be written equivalently using the random closed set $\Sigma$, its price-time parametric representation $(Z,Y)$, or the OHLC processes.
The latter are preferred initially since these facilitate the clearest explanation of how prices compare with the FRH model's.
The other representations certainly deserve their place, as will be explained.

\newcommand{\goesto}{{\ \xrightarrow{a\to\infty}\ }}

\newparahead{Vanillas} Here we explain how to price European call options and digital put options under the FEH model.
Since these depend on a single maturity $T>0$ only, we can consider an FEH payoff that depends on \emph{any one} of the OHLC processes, and their prices will agree.
(Recall that $O_T=H_T=L_T=C_T$ a.s.)
Moreover, resulting prices coincide with the FRH model's, since the FRH process coincides with the FEH model's close process $C$ from \autoref{eq:ohlc}.
In this situation (depending on finite-dimensional distributions only), we define FEH derivative prices in terms of the (c\`adl\`ag) close process $C$ by convention, and leverage Mechkov's existing work (\citeyear{mechkov_2015}).

Recall that $N_t:=\log(\frac{C_t}{s}) - (r-q)t$ has characteristic function $\phi_t(p):=\bb{E}[e^{ip N_t}]=e^{\zeta_N(p)t}$, with NIG exponent $\zeta_N$ as in \autoref{eq:ohlc-char-exponent}.
Then as a consequence of \autoref{thm:fast-excursion-limit} we have the following convergence of Heston prices for strike $K>0$ and maturity $T>0$
\begin{equation*}
    V^\text{Heston}_\text{vanilla-call}(K,T)\ \longrightarrow\ V^\text{FEH}_\text{vanilla-call}(K,T) = V^\text{FRH}_\text{vanilla-call}(K,T),
\end{equation*}
which in mathematical terms should be understood as
\begin{equation}\label{eq:vanilla-call}
    \bb{E}[(S^a_T - K)_+] \goesto \bb{E}[(C_T - K)_+] = F - \frac{1}{\pi}\sqrt{FK}\int_0^\infty \frac{\text{Re}\left[e^{-iuk}\phi_T(u-\frac{i}{2})\right]}{u^2 + \frac14}du.
\end{equation}
Lewis's formula from \citet{lewis_2001} Equation 3.11 is used here, with forward $F:=se^{(r-q)T}$ and log-strike $k:=\log(\frac{K}{F})$, and for convenience these are \emph{forward} prices.\footnote{
    Using deterministic rates, this amounts to ignoring a factor of $e^{-rT}>0.995$ in all reported derivative prices.
The main motivation for this is to align the following derivative price expressions with CDFs / hitting probabilities.
}
For digital puts we similarly have
\begin{equation*}
    V^\text{Heston}_\text{digital-put}(K,T)\ \longrightarrow\ V^\text{FEH}_\text{digital-put}(K,T) = V^\text{FRH}_\text{digital-put}(K,T),
\end{equation*}
which, after using $\bb{E}[1_{A}]=\bb{P}[A]$, this time reads
\begin{equation}\label{eq:digital-put}
    \bb{P}[S^a_T \le K]\goesto\bb{P}[C_T \le K] = \frac{1}{2} - \frac{1}{\pi} \int_{0}^{\infty} \text{Re}\left[\frac{e^{-iuk} \phi_T(u)}{iu}\right] du,
\end{equation}
where here we use \citet{lewis_2001} Equation 3.12, as opposed to integrating the NIG PDF as given in \citet{mechkov_2015}.

\citet{mechkov_2015} can be consulted for example implied volatilities arising from call option prices as in \autoref{eq:vanilla-call}.
We have specifically introduced digital puts as in \autoref{eq:digital-put} because these coincide with marginal CDFs of $C$, which is helpful for interpreting touch option prices as follows.

\newparahead{Touches} Now we look at (down-and-in) single-touch options and (up-and-out) no-touch options,\footnote{
    What we call here a single-touch option is specifically a down-and-in-cash-(at-expiration)-or-nothing option in \citet{haug_2007} Section 4.19.5, and as usual we consider forward values of these.
A no-touch is the corresponding up-and-out.
}
which will conveniently give us two more CDFs to go with that from the digital put in \autoref{eq:digital-put}.

Following the outline for vanillas above, for a single-touch option with strike $K\in(0,s]$ and maturity $T>0$, we rather have the following convergence to an FEH price \emph{above} the FRH price,
\begin{equation*}
    V^\text{Heston}_\text{single-touch}(K,T)\ \longrightarrow\ V^\text{FEH}_\text{single-touch}(K,T) \ge V^\text{FRH}_\text{single-touch}(K,T),
\end{equation*}
which mathematically corresponds to `hitting' probabilities,
\begin{equation}\label{eq:single-touch}
    \bb{P}\left[\min_{t\in[0,T]}S^a_t \le K\right] \goesto \bb{P}\left[\min_{t\in[0,T]}L_t \le K\right] \ge \bb{P}\left[\inf_{t\in[0,T]}C_t\le K\right].
\end{equation}
The inequality of hitting probabilities here simply follows from $L_t\le C_t$ a.s.,
which we expect to be strict in cases of practical interest, e.g.\ excluding the boundary examples in \autoref{sec:parameters}.
We call the gap $V^\text{FEH}_\text{single-touch}-V^\text{FRH}_\text{single-touch}$ between these prices \emph{excursion risk} (or premium), given that the FRH and FEH models differ through instantaneous excursions only.\footnote{This gap will correspond to a \emph{vertical} distance between prices in our figures, e.g.\ \autoref{fig:touches}.
Alternatively a \emph{horizontal} distance can be used to define excursion risk instead, giving a better feel for the `average' excursion size.}

For a no-touch option with strike $K\in[s,\infty)$ and maturity $T>0$, we rather have
\begin{equation*}
    V^\text{Heston}_\text{no-touch}(K,T)\ \longrightarrow\ V^\text{FEH}_\text{no-touch}(K,T) \le V^\text{FRH}_\text{no-touch}(K,T),
\end{equation*}
which mathematically corresponds to `missing' probabilities,
\begin{equation*}
    \bb{P}\left[\max_{t\in[0,T]}S^a_t \le K\right] \goesto \bb{P}\left[\max_{t\in[0,T]}H_t \le K\right] \le \bb{P}\left[\sup_{t\in[0,T]}C_t \le K\right].
\end{equation*}

\newparahead{Barriers} From these expressions it should be clear how to price other barrier derivatives under the FEH model, via the OHLC prices from \autoref{eq:ohlc}.
For some additional clarity, essentially combining the examples above, a double knock-out call option with barriers $K_L<K_H$ has price
\[
\bb{E}\bigg[\mathbbm{1}_{A_T} (C_T - K)_+\bigg],\quad A_T := \left\{K_L < \min_{t\in[0,T]}L_t\right\}\cap\left\{\max_{t\in[0,T]}H_t\le K_H\right\}.
\]

\newparahead{Numerical results} Estimated by simulation, \autoref{fig:touches} shows single-touch and no-touch prices for a range of strikes, and maturity $T=0.08$ (near one month), under both FRH and FEH models.
Digital put prices are also shown for 10 equidistant maturities in $(0,0.08]$, which the models agree on.

The parameters used are broadly indicative of a EURUSD market in 2025.
Specifically:
\begin{equation}\label{eq:eurusd-parameters}
    s=1.16,\quad r=0.05,\quad q=0.02,\quad \sigma=0.07,\quad \rho=0.05,\quad \gamma=0.2.
\end{equation}
\autoref{fig:touches-rho} shows how this changes under an extremely negative correlation regime, more indicative of an equity price, wherein we see the influence of larger downward excursions.

\begin{figure}[ht]
\centering
    \includegraphics[width=0.95\textwidth]{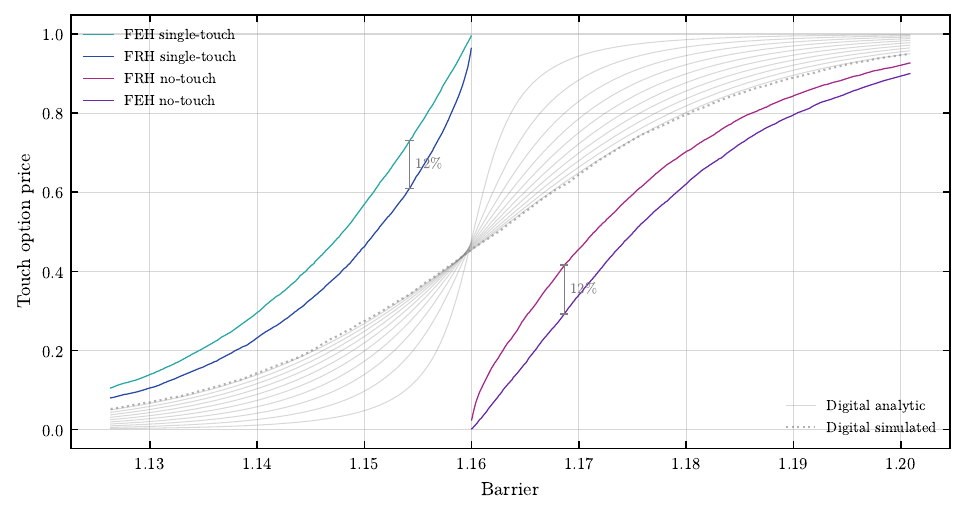}
    \caption{Touch option pricing using parameters from \autoref{eq:eurusd-parameters}.}
    \label{fig:touches}
\end{figure}

\begin{figure}[ht]
\centering
    \includegraphics[width=0.95\textwidth]{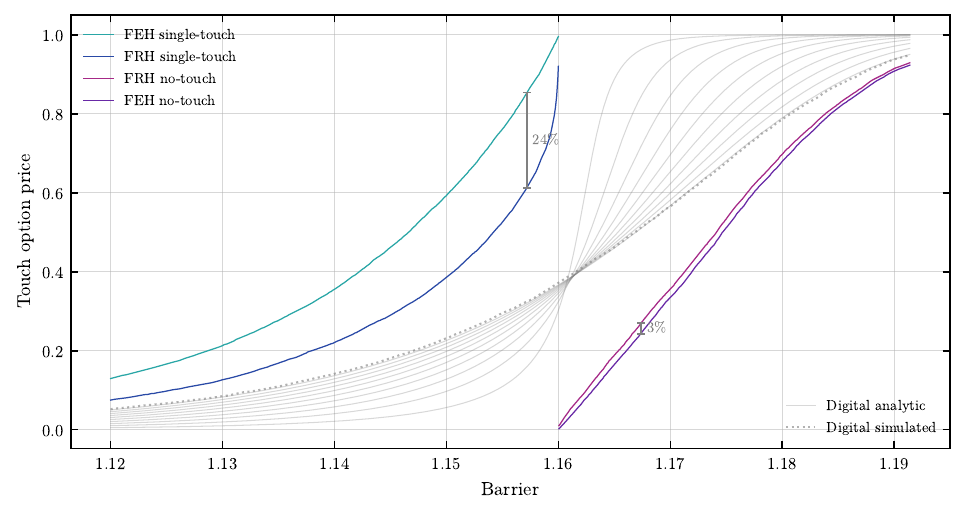}
    \caption{A repeat of \autoref{fig:touches} but using instead $\rho=-0.7$.}
    \label{fig:touches-rho}
\end{figure}

In \autoref{fig:touches-2} we reproduce (approximately) some results from \citet{kudryavtsev_2006} Figure 1.
For this we obtain model parameters $\sigma,\rho,\gamma$ implied by the NIG model, inverting the relationships from \citet{mechkov_2015},\footnote{
    Concretely, assuming the usual constraints $\alpha>|\beta|$, $\alpha>|\beta+1|$ and $\delta>0$, then inverting \citet{mechkov_2015} we arrive at
    \begin{equation*}
        \sigma^2= \frac{\delta^2 + \mu^2}{\delta\sqrt{\alpha^2 - \beta^2}},\quad \rho=-\frac{\mu}{\sqrt{\delta^2 + \mu^2}},\quad \gamma^2 = \frac{1}{\delta\sqrt{\alpha^2 - \beta^2}},
    \end{equation*}
    where $\mu$ is the NIG L\'evy process drift (e.g.\ from \citet{kudryavtsev_2006} Equation 2.2) which preserves $\bb{E}[e^{N_t}]=1$ when it exists, i.e.\
    \begin{equation*}
        \mu = -\delta\left(\sqrt{\alpha^2 - \beta^2} - \sqrt{\alpha^2 - (\beta + 1)^2}\right),
    \end{equation*}
    or equivalently $\mu=-\frac{\sigma\rho}{\gamma}$.
Now using $\alpha=40$, $\beta=6$ and $\delta=11$ from \citet{kudryavtsev_2006}, we arrive at the values in \autoref{eq:kudryavtsev-params}.
} giving
\begin{equation}\label{eq:kudryavtsev-params}
    s=1,\quad r=0.05,\quad q=0,\quad \sigma=0.54,\quad \rho=0.16,\quad \gamma=0.048.
\end{equation}

\begin{figure}[ht]
\centering
    \includegraphics[width=0.95\textwidth]{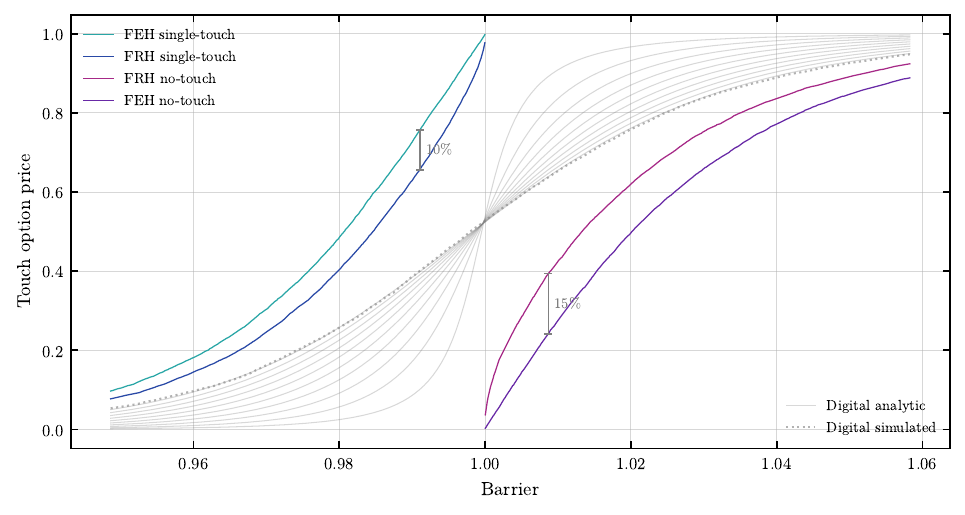}
    \caption{Touch option pricing using parameters from \autoref{eq:kudryavtsev-params}.
Single-touch prices are consistent with the 1 day expiry results from \citet{kudryavtsev_2006} Figure 1.}
    \label{fig:touches-2}
\end{figure}

\autoref{fig:touches} and \autoref{fig:touches-rho} clarify the significance of the excursion risk that is present in the FEH model, showing that it raises hitting probabilities of order 10\%, depending on strike and correlation.

\newparahead{Representations} An FEH price expression in \autoref{eq:single-touch}, in terms of the OHLC processes, feels the most natural when drawing comparisons with classical Heston or FRH counterparts.
But there are many others, e.g.\ for this single-touch we have
\begin{equation}\label{eq:price-representations}
    \bb{P}\bigg[\min_{t\in[0,T]}L_t \le K \bigg] = \bb{P}\bigg[\min\left\{Z_x:Y_x\le T\right\} \le K \bigg] = \bb{P}\bigg[\Sigma\cap \Phi(K,T)\neq\varnothing\bigg].
\end{equation}
Here, along with the first OHLC viewpoint, we have used the price-time parametric representation $(Z,Y)$ instead, and also the random closed set $\Sigma$ with a `hitting set' $\Phi(K,T):=[0,K]\times[0,T]\subset\bb{G}$.

Each of these representations has a well-deserved place.
The final does so because it measures a `hitting family' $\{\Gamma\in\rr{G} : \Gamma\cap \Phi \neq \varnothing \}$, and these are fundamental measurable events in the Effros $\sigma$-algebra $\cc{G}$.
This associated probability (i.e.\ single-touch price) is given by the `capacity functional' of $\Sigma$ evaluated on $\Phi$, and by Choquet's theorem these are sufficient to characterise $\Sigma$.
See \autoref{app:random-closed-sets}.

The central expression deserves its place because it reflects how we estimate this probability, i.e.\ via the simulation of $(Z,Y)$ from \autoref{eq:feh-fluctuation}.
It is straightforward to simulate this process (exactly) e.g.\ on an equidistant grid $x_n=n\delta$ and, mirroring \autoref{thm:fast-excursion-limit}, we get the weak convergence of such a scheme on $(\rr{G},\cc{G})$ as step size diminishes, i.e.\ as $\delta\to0$.
See \autoref{app:simulation}.

\section{Conclusion}

We have introduced a new interval-valued \emph{fast-excursion Heston} (FEH) derivative pricing model, showing how this 1.\ arises from the classical Heston model under Mechkov's fast-reversion limit \citep{mechkov_2015}, and 2.\ raises hitting probabilities / touch-option prices by taking excursion risk into account.

Theoretically, this model provides a rare example of a non-degenerate limit from continuous stochastic volatility models that escapes the Skorokhod topologies, and we are able to explain this limit with \emph{almost sure} results (e.g.\ enabling its visualisation as in \autoref{fig:pathwise-convergence}).\footnote{In this sense our work on limits embodies Skorokhod's representation theorem, as opposed to Prokhorov's theorem \citep{billingsley_1999}.}
In this space where such a.s.\ results are possible, volatility models solve random ODEs as opposed to SDEs.

The FEH model is expressed using the robust theory of random closed sets \citep{molchanov_2017}, in which model trajectories should be thought of as being tremendously simple and intuitive things; merely (closed) sets of price-time points in $\bb{R}^2$.
Within this framework, our set-valued FEH model is a lift of the real-valued fast-\emph{reversion} Heston (FRH) model (an exponentiated NIG L\'evy process) with the FRH selection process being recovered by a projection onto right limits in time.

This model lift can be understood more simply as a generalisation of the composition operator $\circ\rightsquigarrow\bullet$ within the framework of subordinated L\'evy processes.
Concretely, this generalisation relates to set membership:
\[
    \text{FRH model} = Z_{X_t} =: (Z\circ X)_t\ \in\ (Z\bullet X)_t := \{Z_x : x\in[X_{t-},X_t]\} = \text{FEH model}.\footnote{Recall \citet{abijaber_2024} Remark 3.13; this \emph{complete} composition operator $\bullet$ \emph{is} continuous on the Fell topology of $\rr{G}$, unlike $\circ$.}
\]

In the FEH model, $X$ is specifically the IG L\'evy process, but other strictly increasing subordinators (e.g.\ Gamma, $\alpha$-stable) lead to other random closed sets $Z\bullet X$, which induce compact interval-valued processes, and can be considered for derivative pricing just like $Z\circ X$.\footnote{
    To our knowledge, such interval-valued processes $Z\bullet X$, admitting continuous parametric representations $(Z,Y)$ with right inverse $Y:=X^{-1}$, have not been studied.
E.g.\ compare the classic It\^o excursion processes of \citet{bertoin_1996} Section 4.4.
In this setting, $X$ is specifically the \emph{inverse local time of the underlying process} i.e.\ $Z$, as opposed to being free.}
As in \autoref{sec:touch-pricing}, these would agree on vanilla prices because of stochastic continuity, while $Z\bullet X$ would exhibit elevated volatility for exotics due to excursions.

Admittedly it is hard to imagine alternatives being more useful than the FEH model, for two main reasons: 1.\ its three main parameters $\sigma, \rho,\gamma$ are inherited from the classical Heston model and control its implied volatility surface in the same way, and 2.\ the right inverse $Y:=X^{-1}$ is a running maximum of Brownian motion, facilitating exact simulation of the price-time parametric representation $(Z,Y)$.
Hopefully we will be proved wrong on this point.

On the practical side, the FEH model lets us access a certain boundary of classical Heston derivative prices; a boundary where both volatility reversion speed and vol-of-vol have `reached infinity'; a boundary that appears inaccessible to classical PDE and Monte Carlo techniques based on fixed time grids.\footnote{
    E.g.\ if we take the Monte Carlo scheme of \citet{abijaber_2025}, which is practically \emph{designed} to handle this fast-reversion limit (and does so exceedingly well in most respects), hitting probabilities converge instead to the FRH model's as $a\to\infty$.
Such schemes require extending to capture excursions between simulation times, \emph{regardless} of how fine the grid is.
}
Indeed, the simulation schemes used in this article (even for classical Heston), are parametric and as such use adaptive time grids, or \emph{fixed grids in cumulative variance space}.

Moving forward with the FEH model or similar for exotics pricing, a joint simulation scheme for the trio $H_{0,T}$, $L_{0,T}$ and $C_T$ (or some subset) from \autoref{eq:ohlc} would be highly desirable.
This would avoid our simulation scheme using $(Z,Y)$, which is not obviously amenable to vectorisation, and could be considered as an \emph{approximate} simulation scheme for exotics under the \emph{classical} Heston model as well, exactly like that of \citet{abijaber_2025} works for vanillas.

Finally, a comparison of FRH and FEH touch option prices with market quotes will be enlightening.
This will enable us to conclude on whether the excursion risk and exotics premium presented in this article, e.g.\ \autoref{fig:touches}, is a reality in e.g.\ FX or equity markets, which the FEH model or similar may help with.
Such models may be useful in other markets too, e.g.\ to help model and mitigate electricity or interest rate spikes.
E.g.\ compare \autoref{fig:correlation-boundary} with \citet{andersen_2024} Figure 1.

\clearpage
\section*{Acknowledgements}\markboth{Acknowledgements}{Acknowledgements}

The thinking behind this article has taken shape over several years.
It started after being exposed to the FRH model in the Numerix library (\url{www.numerix.com}), around 2015.
This model and the associated Heston limit influenced my PhD, which benefited from Mikko Pakkanen's guidance and correspondence with Eduardo Abi Jaber.
Martin Rasmussen became practically crucial as the project landed in the territory of (probability-free) ODE well-posedness, and I am tremendously grateful to him for extending my visiting status at Imperial College since.
I thank Ward Whitt and more recently Ilya Molchanov for correspondence which has helped with the $\rr{E}$, $\rr{F}$ and $\rr{G}$ spaces.

\section*{Declarations}

\subsection*{Funding}

This work builds on doctoral research supported by the Engineering and Physical Sciences Research Council (EPSRC) Centre for Doctoral Training in Financial Computing and Analytics under Grant EP/L015129/1.
No additional funding was received for the preparation of this article.

\subsection*{Disclosure statement}

The author reports there are no competing interests to declare.
This article represents the views of the author alone and not the views of any firm or institution.

\subsection*{Generative AI use}

Claude Code (Anthropic), using Claude Opus 5, was used for coding assistance and manuscript preparation.
The code that generates this article's figures was developed with it, and it was used for \LaTeX{} tooling, formatting, grammar and spelling checks, and drafting the declarations in this section.
It was not used for the mathematics, the results, or the writing of the article, which are the author's own.
All output was reviewed and revised by the author.

\subsection*{Data availability}

No empirical data were used in this article.
All numerical results are produced by simulation, using the parameters stated in the text.
The code that generates them, including the figures, is openly available at \url{https://github.com/rmcrkd/fast-excursion-limit}.

\phantomsection
\addcontentsline{toc}{section}{References}
\bibliographystyle{abbrvnat}
\bibliography{main}

\clearpage
\begin{appendix}
\phantomsection
\addcontentsline{toc}{section}{Appendices}
\addtocontents{toc}{\protect\setcounter{tocdepth}{2}}
\addtocontents{toc}{\protect\nestappendixtoc}

\section{Background models}\label{app:background-models}

This section provides an overview of existing models which help to see how the FEH model fits in.
The \emph{most} relevant models are the classical Heston model \citep{heston_1993} and NIG L\'evy model \citep{barndorff-nielsen_1994}, which are linked together through Mechkov's FRH model \citep{mechkov_2015}.

All of the models here can be constructed from a fixed probability space $(\Omega,\cc{F},\bb{P})$ that supports a two-dimensional Brownian motion $W=(W^0,W^1)$, and they all leverage a correlated Brownian motion $W^\rho := \sqrt{1 - \rho^2}W^0 + \rho W^1$ where $\rho\in[-1,1]$.
The same notation $S=\{S_t\}_{t\in[0,\infty)}$ is used throughout this section to denote \emph{different} price processes.
All of these models are real-valued and c\`adl\`ag, mostly continuous too.

For the sake of concreteness only, we henceforth imagine that $t=0$ represents the present time, that $\bb{P}$ is not the real-world measure but another under which derivatives will be priced, and that $S$ represents a foreign-exchange rate with fixed spot price $S_0=s\in[0,\infty)$, and fixed domestic and foreign interest rates $r,q\in\bb{R}$.
Here is an example \emph{Black-Scholes} model using these parameters:
\begin{equation}\label{eq:black-scholes}
    S_t = s\exp\left((r - q)t + \sigma W^\rho_t - \frac12 \sigma^2 t \right).
\end{equation}

\subsection{Classical Heston and Mechkov's parameterisation}

We first recall the \emph{classical Heston} model as in \citet{heston_1993},
\[
    \frac{dS_t}{S_t} = (r - q)dt + \sqrt{V_t} dW^\rho_t,\quad dV_t = \kappa(\theta - V_t)dt + \sigma\sqrt{V_t}dW^1_t,\quad (S_0,V_0)=(s,v)\in (0,\infty)^2,
\]
but are keen to move on from this swiftly, in favour of a different parameterisation.
In preparation for a certain limit, \citet{mechkov_2015} introduces a Heston parameterisation that we will call here \emph{Mechkov's Heston} model,
\begin{equation}\label{eq:mechkov-heston}
    \frac{dS_t}{S_t} = (r - q)dt + \sigma\sqrt{V_t} dW^\rho_t,\quad dV_t = a \left((1 - V_t)dt + \gamma\sqrt{V_t}dW^1_t\right).
\end{equation}
The first things to note here are that $\sigma$ has completely changed meaning, now corresponding to the reversion level $\sqrt{\theta}$, and $V$ is now normalised such that e.g.\ $\bb{E}[V_t]\to 1$ as $t\to\infty$.

Broadly we will not care about the starting variance $V_0=v$, and in the fully normalised case of $v=1$ then $\sigma$ represents the `average volatility' of $S$ up to any time.
Indeed the model can be reduced to Black-Scholes with volatility $\sigma$ under both slow and fast-reversion limits.

The non-trivial stability of Mechkov's Heston parameterisation in the limit $a\to \infty$ is clearly what makes it special.
To aid the next transition we note that by solving the SDE for $S$, and writing the SDE for $V$ as an equivalent integral equation, then \autoref{eq:mechkov-heston} becomes
\begin{multline}\label{eq:mechkov-heston-integral}
    S_t = s\exp\left((r - q)t + \sigma \int_0^t\sqrt{V_u}dW^\rho_u - \frac12 \sigma^2 \int_0^t V_u du \right),\\
    V_t = v + a\left(t - \int_0^t V_u du + \gamma\int_0^t \sqrt{V_u}dW^1_u\right).
\end{multline}

\subsection{Time-changed Heston / Heston random ODE}\label{sec:heston-random-ode}

Now we are making our way to a Heston random ODE, the first step of which requires a change of time in the sense of e.g.\ \citet{barndorff-nielsen_2010}, i.e.\ via the Dambis-Dubins-Schwarz theorem \citep{dambis_1965,dubins_1965}.
This e.g.\ allows us to write the It\^o integrals in \autoref{eq:mechkov-heston-integral} in terms of other Brownian motions $B=(B^0, B^1)$ and $B^\rho := \sqrt{1 - \rho^2}B^0 + \rho B^1$.
In particular, we have
\[
M^\rho_t := \int_0^t \sqrt{V_u}dW^\rho_u = B^\rho_{[M^\rho]_t} = B^\rho_{\int_0^t V_udu},
\]
a.s.,\footnote{To be clear, this equation \emph{defines} $B$, e.g.\ $B^0_x:=M^0_{[M^0]^{-1}_x}$ can be shown to be a Brownian motion.} and prioritising $B$ over $W$ in \autoref{eq:mechkov-heston-integral} gives us the \emph{time-changed Heston} model,\footnote{
    While this (time-)change in perspective may feel awkward, it is important to appreciate that Wolfgang Doeblin was working with such equations for diffusions in 1940 \emph{before} the introduction of It\^o integrals and SDEs.
Doeblin died tragically in 1940, with some of his work discovered as late as 2000.
See \citet{bru_2002} and \citet{swishchuk_2016} for more details.
}
\begin{equation}\label{eq:time-changed-heston}
    S_t = s\exp\left((r - q)t + \sigma B^\rho_{\int_0^t V_u du} - \frac12 \sigma^2 \int_0^t V_u du \right),\quad V_t = v + a\left(t - \int_0^t V_u du + \gamma B^1_{\int_0^t V_u du}\right).
\end{equation}

The connection between Brownian motions $W$ and $B$ depends on model parameters, so in order to work with a fixed Brownian motion as we \emph{vary} model parameters (especially as $a\to\infty$) we simply consider \autoref{eq:time-changed-heston} driven by $W$ instead.
Writing $X_t:=\int_0^t V_udu$ so that $X'_t:=\frac{d X_t}{dt} = V_t$ a.s., we arrive at the \emph{Heston random ODE} model
\begin{equation}\label{eq:heston-random-ode}
    S_t = s\exp\left((r - q)t + \sigma W^\rho_{X_t} - \frac12 \sigma^2 X_t \right),\quad X'_t = v + a\bigg(t - X_t + \gamma W^1_{X_t}\bigg),\quad X_0=0.
\end{equation}

Although the mainstream SDE-based models introduced thus far are well understood to have unique strong solutions, with this shift to ODEs things should feel unclear.
An unexpected finding of \citet{mccrickerd_2021} however is that this Heston random ODE has a unique strong solution too, so in particular the distribution of $S$ is equivalent for all models since Mechkov's Heston.\footnote{
    This equivalence is summarised by \citet{mccrickerd_2021} Theorem 4.14.
The resulting random ODE turns out to be incredibly robust.
E.g.\ the corresponding probability-free `Heston ODE' $x' = f(t, x):= v + a\left(t - x + \gamma w_1(x) \right)$ with $x(0)=0$ has a unique solution \emph{for all} continuous $w_1:\bb{R}\to\bb{R}$ starting from zero, due to each $f(\cdot,x)$ being strictly increasing.
}

Regarding these random ODE solutions, a peculiar limiting result of \citet{mccrickerd_2021} was the following, stated here in the form we will use it; in terms of $Y$ as opposed to $X$, and in this article's notation.

\begin{theorem}[Inverse running max limit]\label{thm:running-max-limit}
    For each $a>0$, suppose that $X^a=\{X^a_t\}_{t\in[0,\infty)}$ solves the random ODE in \autoref{eq:heston-random-ode}, and let $Y^a=\{Y^a_x\}_{x\in[0,\infty)}$ be its inverse $Y^a:=(X^a)^{-1}$.
Let $ Y=\{ Y_x\}_{x\in[0,\infty)}$ be defined as the following running maximum
    \begin{equation*}
      Y_x := \max\bigg\{x' - \gamma W^1_{x'}: x'\in[0, x]\bigg\}.
    \end{equation*}
    Then the a.s.\ convergence $Y^a\to Y$ takes place uniformly over compacts as $a\to\infty$.
\end{theorem}

This limit $ Y$ is, in turn, the right inverse of an inverse Gaussian L\'evy process, introduced shortly in \autoref{eq:nig-mean-variance}.
See \citet{mccrickerd_2021} Corollary 4.49 for a more general statement of this result, and Example 3.19 and Figure 12 there for (probability-free) intuition.

\subsection{Fast-reversion Heston / Normal-inverse Gaussian}\label{sec:frh}

The fast-reversion Heston (FRH) model \citep{mechkov_2015} for a price $S$ can be written as $S_t:= s\exp((r-q)t + N_t)$ where $N$ is the L\'evy process with characteristic exponent
\begin{equation}\label{eq:nig-exponent}
    \zeta_N(p) = - \frac{\sigma\rho}{\gamma}ip + \frac{1}{\gamma^2}\left(1 - \sqrt{(1 - i\sigma\rho\gamma p)^2 + (\sigma\gamma)^2(p^2 + ip)}\right),
\end{equation}
i.e.\ with characteristic function $\phi_N(p, t) := \mathbb{E}[e^{ip N_t}]=e^{\zeta_N(p)t}$.
This makes $N$ a normal-inverse Gaussian (NIG) L\'evy process with drift such that $\phi_N(-i,t)=\mathbb{E}[e^{N_t}]=1$.
The corresponding `standard' NIG distribution / process parameters $\alpha,\beta,\delta,\mu$ are given in \citet{mechkov_2015}, but we stress the relative value of this parameterisation in volatility modelling, given that $\sigma,\rho,\gamma$ preserve their interpretation from the Heston model.

Using the NIG process's representation by subordination \citep{barndorff-nielsen_1997}, we can construct (a c\`adl\`ag version of) this model in terms of $W$ by using
\begin{equation}\label{eq:nig-mean-variance}
     N_t = \sigma\sqrt{1-\rho^2} W^0_{X_t} + \frac{\sigma\rho}{\gamma}(X_t - t) - \frac12\sigma^2  X_t,\quad  X_t := \inf\left\{x > 0 : x - \gamma W^1_x > t\right\}
\end{equation}
where we recognise $X$ as the inverse Gaussian (IG) L\'evy process,
with characteristic exponent
\begin{equation}\label{eq:ig-exponent}
    \zeta_X(p) := \frac{1}{\gamma^2}\left(1 - \sqrt{1 - 2\gamma^2ip}\right).
\end{equation}
It turns out that this model for $S$ can be written exactly as we have done in \autoref{eq:heston-random-ode}, i.e.\
\begin{equation}\label{eq:fast-reversion-heston}
    S_t = s\exp\left((r - q)t + \sigma W^\rho_{X_t} - \frac12 \sigma^2 X_t \right),
\end{equation}
where $X$ is instead the IG L\'evy process here.
This can be obtained after noting that $X$ satisfies $X_t - t = \gamma W^1_{X_t}$.
Now the main result of \citet{mechkov_2015} was to establish that the marginal distributions of the Heston model parameterised as in \autoref{eq:mechkov-heston} converge to those of this FRH model as $a\to\infty$.
A key insight from \citet{mccrickerd_2021} was the understanding that this originates from a deeper connection between the integrated Cox-Ingersoll-Ross process $\int_0^t V_udu$ and IG L\'evy process.
This connection is summarised in the Epilogue of \citet{mccrickerd_2021}, has been studied also in \citet{abijaber_2024}, and applied effectively in \citet{abijaber_2025}.

\section{Random closed sets}\label{app:random-closed-sets}

This appendix provides an introduction to random closed sets, which is the natural framework in which to define the FEH model and establish it as a Heston model limit.
This introduction comes almost exclusively from Molchanov's text (\citeyear{molchanov_2017}), with some probability-free topics from \citet{beer_1993}.
As will become clear, the presentation here has applications to derivative pricing in mind.

In general the realisations of random closed sets are subsets of an underlying topological space; the `carrier' space, denoted here by $\bb{G}$.\footnote{
    The carrier space is denoted by $\bb{E}$ in \citet{molchanov_2017}.
    We change notation like this to avoid some horrendous clashes.
}
In \autoref{sec:fast-excursion-heston} this carrier space is our price-time domain, the Euclidean space $\bb{G}:=\bb{R}\times[0,\infty)$, with time zero representing `now'.
Consider the realisations or graphs in \autoref{fig:pathwise-convergence}.
With $\bb{G}$ defined as such, or as any $\bb{R}^d$, we can take for granted that $\bb{G}$ is \emph{locally compact Hausdorff second countable} (LCHS); see \citet{molchanov_2017} Chapter 1.
This facilitates a streamlined presentation towards the relevant Effros $\sigma$-algebra $\cc{G}$, on which our model's distribution is defined.

The family of \emph{all} such closed subsets of $\bb{G}$ is denoted by $\rr{G}$, for graphs.
For the avoidance of doubt, $\rr{G}$ \emph{does} contain $\varnothing$.
Such notation introduced in this appendix is summarised in \autoref{tab:app-notation}.

\begin{table}[ht]
\begin{center}
\begin{tabular}{ll}
 Symbol & Description \\ [0.5ex]
 \hline
 $\bb{G}$ & Price-time domain $\bb{R}\times[0,\infty)$ \\
 $\rr{G}$ & Family of closed subsets of $\bb{G}$ \\
 $\Gamma$ & An element of $\rr{G}$; a graph \\
 $\cc{S}$ & Sub-base of the Fell topology on $\rr{G}$ \\
 $\cc{T}$ & Fell topology on $\rr{G}$ \\
 $d_\text{AW}$ & Attouch-Wets metric on $\rr{G}\setminus\varnothing$ \\
 $\cc{G}$ & Effros $\sigma$-algebra on $\rr{G}$ \\
 $\Sigma$ & Random closed set on $(\rr{G},\cc{G})$ \\
 $\Pi_\Sigma$ & Capacity functional of $\Sigma$
\end{tabular}
\end{center}
\caption{Main notation introduced in \autoref{app:random-closed-sets}, most of which is adopted in \autoref{sec:fast-excursion-heston}.}
\label{tab:app-notation}
\end{table}

This introduction also serves the case for using random closed sets.
Besides topological robustness, ideal for limit theorems, this is largely based around the enlargements of spaces $\rr{C}\to\rr{D}\to\rr{E}\to\rr{F}\to\dots$ sketched out in \citet{whitt_2002}.\footnote{See in particular the summary at the end of \citet{whitt_2002} Section 15.3.}
To be clear, our space $\rr{G}$ of graphs \emph{does not} reside at the end of this chain. $\rr{G}$ is rather an alternative to Whitt's excursion space $\rr{E}$,
with a fraction of its elements being conveniently induced by those of his fluctuation space $\rr{F}$, precisely as in \autoref{def:fast-excursion-heston}.

\newparahead{Hit-and-miss families}
Given a subset $\Phi\subset\bb{G}$, not necessarily closed and therefore not necessarily \emph{in} $\rr{G}$,
various topologies \emph{on} $\rr{G}$ are built from the following subfamilies of graphs $\Gamma\in\rr{G}$,
\begin{equation*}
    \rr{G}_\Phi:= \{\Gamma\in\rr{G} : \Gamma\cap \Phi \neq \varnothing \},\quad \rr{G}^\Phi:= \{\Gamma\in\rr{G} : \Gamma\cap \Phi = \varnothing \}.
\end{equation*}
These are called the hit and miss families of $\Phi$ in $\rr{G}$ respectively, for fairly obvious reasons.

It is remarkable that from this outset we can draw connections with derivative pricing.
E.g.\ we can think of it as being the job of our model to make graphs $\Gamma$ like in \autoref{fig:pathwise-convergence}, and the job of a derivative to define a payoff region $\Phi$.
Then e.g.\ the hit subfamily $\rr{G}_\Phi\subset\rr{G}$ contains the graphs $\Gamma$ for which this derivative, represented by $\Phi$, is in-the-money.
See how this works in \autoref{eq:price-representations}.

\newparahead{Fell topology}
Various so-called hit-and-miss topologies can be constructed on $\rr{G}$ by generating them from hit-and-miss families.
Following \citet{molchanov_2017}, we focus on the Fell topology $\cc{T}$, which is specifically generated from the sub-base $\cc{S}$ of \emph{hit} families $\rr{G}_\Phi$ for \emph{open} $\Phi$ and \emph{miss} families $\rr{G}^\Phi$ for \emph{compact} $\Phi$.\footnote{
    Recall this means $\cc{T}$ is the smallest topology containing $\cc{S}$ as open sets, constructed by first taking all \emph{finite} intersections of $\cc{S}$ (forming a base) then all unions (forming a topology).
The specific combinations of hit-and-miss families that generate hit-and-miss topologies like Fell's is beyond us.
    However see \citet{beer_1993} which considers the finer Vietoris hit-and-miss topology as the prototype, using \emph{closed} miss families in place of compact in \autoref{eq:fell-sub-base}.
}
Writing $\tau(\cc{S})$ for the topology generated by sub-base $\cc{S}$, this concretely means
\begin{equation}\label{eq:fell-sub-base}
    \cc{T}:=\tau(\cc{S}),\quad \cc{S}:=\ \big\{\ \rr{G}_\Phi:\Phi\ \text{open}\ \big\}\ \cup\ \big\{\rr{G}^\Phi:\Phi\ \text{compact}\big\}.
\end{equation}

As discussed in \autoref{sec:fast-excursion-heston}, the Fell topology $\cc{T}$ on $\rr{G}$ is tremendously robust.
It is e.g.\ \emph{compact} (despite $\bb{G}$ not being) and metrisable.\footnote{
    Comparables are Skorokhod's $\rr{M}_2$ topology on $\rr{D}$, especially constructions over $[0,\infty)$ as in \citet{billingsley_1999}, or Whitt's counterpart on $\rr{E}$, as in \citet{whitt_2002}.
These \emph{incomplete} spaces contain `special realisations' in $\rr{G}$, as elegantly put in \citet{molchanov_2017}.
}
The Fell topology on the subset $\rr{G}':=\rr{G}\setminus\varnothing$ is instead `just' Polish and metrisable.
See \citet{beer_1993} Theorem 5.1.5, which consolidates these properties.

Since topologies like $\cc{T}$ are constructed on a set like $\rr{G}$ of \emph{subsets} in an underlying space like $\bb{G}$, they are called `hyperspace' topologies, with properties of $\bb{G}$ having direct consequences for $\cc{T}$.

\newparahead{Attouch-Wets metric} This Fell topology on $\rr{G}'$ is specifically induced by the \emph{Attouch-Wets} metric, due to $\bb{G}$ having the Heine-Borel property; see \citet{beer_1993} Exercise 5.1.10.
This proves convenient for establishing convergence, as shown in \autoref{thm:fast-excursion-limit}.
For $\Gamma_1, \Gamma_2\in\rr{G}'$, this metric is defined by
\[
    d_\rr{AW}(\Gamma_1,\Gamma_2) := \sum_{n=1}^\infty \frac{1}{2^n} \left(1\wedge \sup_{a\in B_n} \left|d(a, \Gamma_1) - d(a, \Gamma_2)\right|\right),
\]
where $B_n:=\{a\in\bb{G}:\Vert a\Vert<n\}$ are balls of radius $n$, and $d(a,\Gamma):=\inf_{b\in\Gamma}d(a,b)$ is the Euclidean distance from point to set.

The Attouch-Wets and Hausdorff metrics are evidently closely related, and induce the same (hyperspace) topologies for \emph{bounded} carrier spaces $\bb{G}$.
See \citet{beer_1993} Chapter 3 for a thorough comparison.

\newparahead{Effros $\boldsymbol{\sigma}$-algebra} In general, the Effros $\sigma$-algebra $\cc{G}$ is defined as that generated from just the \emph{open} hitting sets in \autoref{eq:fell-sub-base}.
However, with our carrier space $\bb{G}$ being LCHS, \emph{$\cc{G}$ coincides with the Borel $\sigma$-algebra $\cc{B}(\cc{T})$ generated by the Fell topology $\cc{T}$ on $\rr{G}$}.
See \citet{molchanov_2017} Section 1.3.3.

The Fell topology itself is not especially helpful in \autoref{sec:fast-excursion-heston}, so we prefer the short Effros description of
$\cc{G}$.
This Effros-Borel coincidence $\cc{G}=\cc{B}(\cc{T})$ should not be forgotten, however, because combined with the metrisability of $\cc{T}$, this brings our space $(\rr{G},\cc{G})$ into the comfort zone of \citet{billingsley_1999}.

\newparahead{Random closed sets} The cleanest definition of a random closed set is simply a random element of the space $(\rr{G},\cc{G})$ in the familiar sense of \citet{billingsley_1999}.
That is, \emph{a random closed set is a measurable map}:
\[
    \Sigma:(\Omega,\cc{F},\bb{P})\to (\rr{G},\cc{G}),
\]
which ensures a well-defined distribution $\mu_\Sigma:=\bb{P}\Sigma^{-1}$ on $(\rr{G},\cc{G})$.
This definition is equivalent to \citet{molchanov_2017} Definition 1.1.1'' for any LCHS carrier space $\bb{G}$, because of the resulting coincidence $\cc{G}=\cc{B}(\cc{T})$.\footnote{
    This random closed set definition for LCHS $\bb{G}$ therefore also coincides with the less familiar Definition 1.1.1' and Definition 1.1.1 from \citet{molchanov_2017}.
And in our specific case $\bb{G}=\bb{R}\times[0,\infty)$, we should think of $\Sigma$ rather as a \emph{random graph} with trajectories $\Gamma$ like in \autoref{fig:pathwise-convergence}.
But since `graph' has many meanings, \emph{random closed set in $\rr{G}$} is preferred.
}

The next few parts of this appendix clarify how derivatives can be priced in this framework of random closed sets, and how other models (e.g.\ It\^o SDE or L\'evy-based) can be lifted into it.

\newparahead{Digital options}
Some of the simplest (and most liquid) `digital' derivatives pay $1$ (say USD) if a price satisfies certain conditions, and $0$ otherwise.\footnote{
    By price here, we really mean a set of prices as in \autoref{fig:pathwise-convergence}, usually limited by a finite time horizon $T<\infty$.
}
E.g.\ a digital put option pays $1$ if a price is below a strike $K$ at a \emph{specific} time $T$, and a (down-and-in) single-touch option pays $1$ if a price falls below $K$ at any time \emph{before} $T$.
Both of these examples are considered in \autoref{sec:touch-pricing}.

Such conditions (thus derivatives) may be represented by hitting sets $\Phi$, with the associated payoff being represented through events like $\{\Sigma\cap \Phi\neq\varnothing\}=\{\Sigma\in\rr{G}_\Phi\}$.\footnote{
    E.g.\ see how the specific case of $\Phi:=[0,K]\times[0,T]\subset\bb{G}$ appears in \autoref{eq:price-representations} for a single-touch option.
}
Concretely, under a model $\Sigma$, the \emph{random} payoff $H$ and \emph{expected} payoff $\Pi$ (both in $\bb{R}$) are given respectively by
\begin{equation}\label{eq:capacity}
    H_\Sigma(\Phi) = \mathbbm{1}_{\Sigma\cap \Phi\neq\varnothing},\quad \Pi_\Sigma(\Phi) = \bb{E}[H_\Sigma(\Phi)] = \bb{P}[\Sigma\cap \Phi\neq\varnothing].
\end{equation}

Given how closely related derivative pricing is to expected payoffs (typically), we can think of $\Pi$ as a pricing function for such options.\footnote{
    We are neglecting matters relating to payoff timings and associated discounting.
In \autoref{sec:touch-pricing} we specifically consider \emph{forward} prices, to preserve connections between prices and probabilities, making \autoref{fig:touches} more meaningful.
}
\autoref{sec:touch-pricing} clarifies how \emph{other} derivatives can be priced, but our specific agenda there is to provide formulas that are consistent with a Heston model limit.

More generally, the \emph{fundamental measurability theorem} brings reassurance in abundance regarding the potential of random closed sets for derivative pricing.
See \citet{molchanov_2017} Theorem 1.3.3.

\newparahead{Capacity functionals} Constraining $\Phi$ to compacts only, the `digital option pricing' map $\Pi_\Sigma$ from \autoref{eq:capacity} is called the \emph{capacity functional} of $\Sigma$, and Choquet's theorem tells us that these \emph{characterise} random closed sets since $\bb{G}$ is LCHS. See \citet{molchanov_2017} Definition 1.1.17 and Theorem 1.1.29.

In more practical terms: \emph{digital option prices $\Pi(\Phi)$ for hitting sets $\Phi$ characterise models}.
So if we had market prices for all such options, then either these would define an admissible capacity functional, and therefore uniquely identify a model (which could then be used for a range of other tasks), or they would not be admissible and e.g.\ arbitrage may be identifiable as a result.\footnote{
    See Choquet's theorem regarding admissibility; $\Pi$ must be upper semicontinuous and completely alternating.
}

It is worth noting that while capacity functionals $\Pi_\Sigma$ look like probability measures, the two coincide iff $\Sigma$ is a random singleton, but otherwise $\Pi_\Sigma$ is non-additive.
See \citet{molchanov_2017} Proposition 1.1.30.

\newparahead{Processes} Clearly we can think of a random closed set $\Sigma$ on our price-time domain $\bb{G}=\bb{R}\times[0,\infty)$ as a (closed) set-valued stochastic process $\{\Sigma_t\}_{t\in[0,\infty)}$ of prices evolving over time.
\citet{molchanov_2017} Chapter 5 is dedicated to such links between random sets and stochastic processes, see Section 5.1.2 especially.

We avoid details here because for us thinking of $\Sigma$ as a \emph{set}-valued process is merely a convenience that helps to clarify relationships with the familiar \emph{real}-valued counterparts of \autoref{app:background-models}.\footnote{
    E.g.\ we will not rely on properties like joint measurability or martingality of processes as in \citet{molchanov_2017} Chapter 5.
}
This will be abundantly clear in our application to derivative pricing, e.g.\ in \autoref{eq:price-representations} we give three helpful representations of a single-touch option price, none of which feature a set-valued process.

\newparahead{Selections}
When clarifying process relationships, the concept of a \emph{selection process} is important.
In our setting, a \emph{real}-valued process $S=\{S_t\}_{t\in[0,\infty)}$ is a selection process of our \emph{compact interval}-valued process $\Sigma=\{\Sigma_t\}_{t\in[0,\infty)}$ if $S_t\in\Sigma_t$ a.s.\ for each $t\in[0,\infty)$.
See \citet{molchanov_2017} Definition 5.1.22.

The next two parts summarise how we will project from our model $\Sigma$ onto \emph{several} real-valued selection processes (used in \autoref{sec:ohlc-processes} especially), and how converse lift operations work also.

\newparahead{Projections}
Viewing a random closed set $\Sigma$ as a set-valued process $\{\Sigma_t\}_{t\in[0,\infty)}$ as above can be understood as a family of projections onto time-slices of $\Sigma$.
See related \citet{molchanov_2017} Proposition 1.3.26.

For our application in \autoref{sec:touch-pricing}, we never need the \emph{whole} slice $\Sigma_t$ of prices; we are just interested in certain extremes.
E.g.\ we will consider a `high' price process $H=\{H_t\}_{t\in[0,\infty)}$ with $H_t:=\max\Sigma_t$, and a `close' process $C_t=\{C_t\}_{t\in[0,\infty)}$ arising via right limits in time, satisfying $\Sigma_{t+}:=\lim_{T\downarrow t}\Sigma_T=\{C_t\}$.
Both of these are convenient \emph{real}-valued projections, but neither is guaranteed to exist.\footnote{
    These processes beg the question of whether random closed sets can be avoided, focussing on `feature vectors' of real processes instead, see \citet{hess_1999}.
A process like $H$ being simpler than $\Sigma$ is illusory, however, indeed $(\rr{G}, \cc{G})$ is a natural state space for (a lift of) $H$.
When possible, focussing on \emph{parametric representations} would seem more fruitful.
}

\newparahead{Lifts}
Converse to projections, we can lift compatible real processes into $(\rr{G},\cc{G})$ through the concept of completed graphs, which are fundamental for the Skorokhod topologies.
See \citet{whitt_2002} Section 3.3.

For a \emph{continuous} process $S=\{S_t\}_{t\in[0,\infty)}$, the lift is given trivially by $\Sigma:=\{(S_t,t):t\in[0,\infty)\}$ as in \autoref{thm:fast-excursion-limit}, which is equivalent to setting random \emph{singletons} $\Sigma_t=\{S_t\}$.
Contrasting this, our process $H=\{H_t\}_{t\in[0,\infty)}$ from \autoref{sec:ohlc-processes} has both right and left limits but is neither c\`adl\`ag nor c\`agl\`ad.
In this case the lift of $H$ into $(\rr{G},\cc{G})$ is obtained by setting $\Sigma_t:=[H_{t-}\wedge H_{t+},H_t]$.\footnote{
    While it is not obvious that the resulting process $\Sigma=\{\Sigma_t\}_{t\in[0,\infty)}$ has \emph{closed} trajectories $\Gamma\in\rr{G}$ a.s., it is straightforward to show this by flooring the price component $Z$ in the parametric representation from \autoref{eq:feh-fluctuation}.
}

\newparahead{Concluding remarks}
This appendix provides an introduction to the theory of random closed sets, with full focus on topics that we rely on, relating to stochastic processes in finance, their limits and applications to derivative prices.
Many other topics will be found in \citet{molchanov_2017}.
Perhaps the Choquet and Aumann integrals, or set-valued martingales were the most narrowly avoided here.

\section{Select proofs}\label{app:proofs}

This appendix proves two convergence results.
The first \autoref{lem:attouch-wets} establishes continuity of a map from price-time parametric representations like $(Z,Y)$ to their graphs w.r.t.\ the Attouch-Wets metric.
This is relied upon in the proof of \autoref{thm:fast-excursion-limit}, but is also used to obtain convergence of our FEH model simulation scheme provided in \autoref{app:simulation}.
The second \autoref{lem:ode-convergence} is very niche.
This is only relied upon for convergence of our Heston random ODE simulation scheme, which is purpose-built for \emph{visualisation} of \autoref{thm:fast-excursion-limit}, as in \autoref{fig:pathwise-convergence}.

\newparahead{Attouch-Wets convergence}
In \autoref{thm:fast-excursion-limit} we pass from the convergence $(\hat Z^a, \hat Y^a)\to(\hat Z,\hat Y)$ of parametric representations to the convergence $\hat \Sigma^a\to\hat \Sigma$ of induced random closed sets.
This passage is justified by the following probability-free result, which applies to our FEH model a.s.

\begin{lemma}\label{lem:attouch-wets}
    Suppose that each $\Gamma_n\in\rr{G}$ is the graph of continuous parametric representation $\gamma_n=(\sigma_n,\tau_n)$ with $\tau_n(0)=0$, $\tau_n$ non-decreasing and unbounded, i.e.
    \[
    \Gamma_n = \{\gamma_n(x)\in\bb{G}:x\in[0,\infty)\}.
    \]
    If $\gamma_n\to\gamma_0$ uniformly over compacts as $n\to\infty$, then $d_\rr{AW}(\Gamma_n,\Gamma_0)\to 0$, and therefore $\Gamma_n\to\Gamma_0$ on the Fell topology of $\rr{G}$ as well.
\end{lemma}

\begin{proof} To obtain $d_\rr{AW}(\Gamma_n,\Gamma_0)\to 0$ it suffices to show that for any radius $k>0$ and error $\varepsilon>0$, we can find $N>0$ such that
\begin{equation}\label{eq:aw-criterion}
\sup_{a\in B_k}\left|d(a, \Gamma_n) - d(a, \Gamma_0)\right| < \varepsilon\quad \forall n>N,
\end{equation}
and convergence on the Fell topology as stated follows immediately.
Now fix some $x_*$ such that $\tau_0(x_*) > 2k + |\sigma_0(0)|$.
Then $d(a, \Gamma_0)\le d(a,\gamma_0(0))<d(a,\gamma_0(x))$ for all $x>x_*$, so we obtain
\begin{equation}\label{eq:nearest-point}
d(a, \Gamma_0) = \min_{x\in[0,x_*]}d(a, \gamma_0(x))=: d(a,\gamma_0(x_0)),
\end{equation}
where continuity guarantees some $x_0\in[0,x_*]$, not necessarily unique.
Using $\sigma_n(0)\to\sigma_0(0)$ and $\tau_n(x_*)\to \tau_0(x_*)$, we can take $n$ up such that \autoref{eq:nearest-point} holds for each $\Gamma_n$, $\gamma_n$ and $x_n\in[0,x_*]$ for all $n$ sufficiently high.
The triangle inequality and assuming $\Vert \gamma_0 - \gamma_n\Vert_{[0,x_*]}<\varepsilon$ then provides
\[
d(a,\Gamma_n) \le d(a,\gamma_n(x_0)) \le d(a,\gamma_0(x_0)) + d(\gamma_0(x_0), \gamma_n(x_0)) <  d(a,\Gamma_0) + \varepsilon.
\]
By symmetry $d(a,\Gamma_0) < d(a,\Gamma_n) + \varepsilon$ for all $n$ high enough too.
These arguments apply uniformly over $a\in B_k$, providing the claim.
\end{proof}

\newparahead{ODE solution convergence} With the next result we clarify the a.s.\ convergence $Y^\varepsilon\to Y$ as $\varepsilon\to0$ uniformly over compacts, as used in \autoref{app:simulation}, specifically relating to the `inverse' ODEs in \autoref{eq:inverse-ode}.
We do this by considering a single Brownian path $w_1:=W^1(\omega)$,\footnote{
    We just need $w_1$ continuous over $[0,\infty)$ with $w_1(0)=0$ and $y(x):=\max_{x'\in[0,x]}\{x' - \gamma w_1(x')\}\to\infty$ as $x\to\infty$.
}
and studying the following probability-free Heston ODE in place of the random counterpart from \autoref{eq:heston-random-ode}
\begin{equation*}
    x' = f(t, x):= v + a\left(t - x + \gamma w_1(x) \right),\quad x(0)=0.
\end{equation*}

\begin{lemma}\label{lem:ode-convergence}
    Assume $f\in\rr{C}(\bb{R}^2,\bb{R})$, each $f(\cdot,x)$ is strictly increasing and $f(0,0)>0$, and let $\varphi\in\rr{C}^1_0([0,\infty),\bb{R})$ be the unique solution of the IVP $x'=f(t,x)$, $x(0)=0$.
Define $f_\varepsilon:= f\vee \varepsilon$ for each $\varepsilon > 0$.
Then the IVPs $x'=f_\varepsilon(t,x), x(0)=0$ also have unique solutions, $\varphi_\varepsilon\in\rr{C}^1_0([0,\infty),\bb{R})$.
Moreover, we have $\varphi_\varepsilon\to \varphi$ and therefore also $\varphi^{-1}_\varepsilon\to \varphi^{-1}$ uniformly over compacts as $\varepsilon\to 0$.
\end{lemma}
\begin{proof}
    Firstly $\varphi$ is unique by \citet{mccrickerd_2021} Theorem 2.17.
The proof of that result is readily adapted for each $f_\varepsilon$.
But since $f_\varepsilon$ is strictly positive, we can use Wend's theorem instead, giving us uniqueness of $\varphi_\varepsilon$.
See \citet{mccrickerd_2021} Theorem 2.14.
For the continuity claim $\varphi_\varepsilon\to \varphi$, we consider also the case of $\varepsilon=0$ and will show that $\varphi_\varepsilon\to\varphi_0= \varphi$.
    So we now focus on this harder case of $\varepsilon=0$.
Recall that \citet{mccrickerd_2021} Lemma 2.6 gives us the bounds $0\le \varphi\le\overline\varphi$ where $\overline\varphi$ is the c\`adl\`ag path defined by
    \[
    \overline\varphi(t) := \inf\{x>0: f(t,x)<0\}.
    \]
    The proof of that result does not really need modifying to obtain $0\le \varphi_0\le\overline\varphi$ also.
The mean-value theorem just leads us to the milder contradiction of $f_0(t,\varphi_0(t))=0<\varphi'_0(t)$, compared with their $f(t,\varphi(t))<0<\varphi'(t)$.
This bounds $\varphi$ and $\varphi_0$ into a region where $f$ and $f_0$ agree, so $\varphi_0=\varphi$, and $\varphi_0$ is unique too.
Since all the IVPs mentioned are thus well-posed, we get a continuity result like \citet{mccrickerd_2021} Theorem 2.18.
I.e.\ the claim of $\varphi_\varepsilon\to\varphi_0=\varphi$ follows from $f_\varepsilon\to f_0$ uniformly over compacts as $\varepsilon\to0$. $\varphi_\varepsilon^{-1}\to\varphi_0^{-1}=\varphi^{-1}$ then follows uniformly over compacts as $\varepsilon\to0$ by continuity.
\end{proof}

\newpara Recall that in \autoref{app:simulation}, $Y$ is the \emph{inverse} of Heston random ODE solution $X$.
So this convergence $\varphi_\varepsilon^{-1}\to\varphi^{-1}$ from \autoref{lem:ode-convergence} gives us $Y^\varepsilon\to Y$ as $\varepsilon\to0$ a.s., as required by using \autoref{eq:inverse-ode}.

\section{Simulation details}\label{app:simulation}

In this appendix we clarify the simulation schemes used for the trajectories in e.g.\ \autoref{fig:pathwise-convergence}.
The FEH model scheme is straightforward compared with our \emph{classical} Heston scheme, and so, ironically, most of the space here is dedicated to the latter, which just serves visualisation purposes.

By design, and like the scheme of \citet{abijaber_2025}, our classical scheme is stable under Mechkov's fast-reversion limit.
There is however a critical difference, highlighted by the development and persistence of excursions in \autoref{fig:pathwise-convergence} as $a\to\infty$, and the resulting convergence of our scheme to trajectories of the FEH model, precisely as \autoref{thm:fast-excursion-limit} says.

This is achieved through the use of random ODEs and price-time parametric representations, which work together to give us an adaptive simulation grid over time, or fixed grid in cumulative variance space.
We believe these techniques are novel, but our specific goal here is visualisation; we have few reasons to believe these can be useful in practice, e.g.\ compared with extensions of \citet{abijaber_2025}.

The following schemes are implemented in the simulate method of \href{https://github.com/rmcrkd/fast-excursion-limit/blob/main/fast_excursion_limit/fast_excursion_heston.py}{FastExcursionHeston} and \href{https://github.com/rmcrkd/fast-excursion-limit/blob/main/fast_excursion_limit/heston_random_ode.py}{HestonRandomODE} respectively.

\newparahead{Fast-excursion Heston}
To simulate the FEH model as in \autoref{def:fast-excursion-heston}, we utilise the price-time parametric representation from \autoref{eq:feh-fluctuation} which defines it, i.e.\
\begin{equation*}
     Z_x = s\exp\bigg((r - q) Y_x +  \sigma W^\rho_x - \frac12 \sigma^2x\bigg),\quad
     Y_x = \max\bigg\{x' - \gamma W^1_{x'}: x'\in[0, x]\bigg\}.
\end{equation*}
Then when it comes to pricing derivatives based on this, we leverage expressions like the following from \autoref{eq:price-representations} for a single-touch option
\[\bb{P}\bigg[\min\left\{Z_x:Y_x\le T\right\} \le K \bigg].\footnote{Having essentially truncated trajectories here at maturity $T$, convergence of the simulation scheme and prices relies on stochastic continuity, i.e.\ there being zero probability of an excursion at the time $T$, i.e.\ $Y^{-1}_{T-}=Y^{-1}_T$ a.s.}\]

It is straightforward to simulate the Brownian motion $(W^0,W^1)$ on a fixed grid $x_n=x_{n,\delta}:= n\delta$ for step size $\delta>0$ and $n\in\bb{N}_0$, with convergence being uniform over compacts as $\delta\to0$.
Convergence of $Y$ and $Z$ follow similarly, by continuity of the maps involved.
Indeed $Y$ (and thus $Z$) can be simulated exactly at these points $x_n$ through the use of Brownian bridges, and we do not need to stipulate interpolation schemes between these (closed sets of) price-time points $(Z_{x_n},Y_{x_n})$, in order to get a.s.\ convergence w.r.t.\ the Attouch-Wets distance and on the Fell topology of $\rr{G}$.

\newparahead{Heston random ODE}
Our starting point here is the following time-changed Heston price process $S$ from \autoref{sec:heston-random-ode}, written explicitly in terms of a (well-posed) random ODE solution $X$,
\[
    S_t = s\exp\left((r - q)t + \sigma W^\rho_{X_t} - \frac12 \sigma^2 X_t \right),\quad \frac{d X_t}{dt} = a\left(t - X_t + \gamma W^1_{X_t}\right) + v =: f(t, X_t),\quad X_0=0.
\]

There will be issues with simulating $S$ on a fixed time grid as $a\to\infty$, as we will inevitably lose information in composed processes like $W^\rho_X$ as $X$ `develops jumps'.
We therefore consider working with a fixed \emph{space} grid for the inverse $X^{-1}$.
Implicitly, this will give us a \emph{random} time grid for $X$ that naturally adapts to its gradient, preventing such information loss, even as $a\to\infty$.

In general, such implicit random time grids are generated by working with a fixed grid on some process $Y$ and parametric representation $(S_{Y}, Y)$ of $S$.
In our case we choose the strictly increasing process $Y:= X^{-1}$, meaning that we consider simulating $(Z, Y)$ where
\[
    Z_x := S_{Y_x} = s\exp\left((r - q)Y_x + \sigma W^\rho_x - \frac12 \sigma^2 x \right).
\]
From such a parametric representation $(Z, Y)$, we recover standard Heston prices using $S=Z_X$.

With a view to working with (approximations of) $Y$ and $Z$, we note that this process $Y$ solves the random ODE
$Y'_x = \frac{1}{f(Y_x, x)}$ whenever $Y'$ exists.\footnote{
    This follows from the inverse function theorem and Lebesgue's theorem.
    See \citet{mccrickerd_2021} Lemma 2.16.
}
Clearly there may be issues trying to work with such an ODE directly, as $f$ approaches zero.
To sidestep this, we define $f_\varepsilon:= f\vee\varepsilon$ for $\varepsilon>0$, let $X^\varepsilon$ solve ODEs driven by $f_\varepsilon$ and define $Y^\varepsilon:=(X^\varepsilon)^{-1}$.
This process solves the simpler problem,
\begin{equation}\label{eq:inverse-ode}
    \frac{d Y^\varepsilon_x}{dx} = \frac{1}{f_\varepsilon(Y^\varepsilon_x, x)},\quad Y^\varepsilon_0=0.
\end{equation}
Making effective use of this random ODE rests on a convergence result as $\varepsilon\to0$, see \autoref{lem:ode-convergence}.

Now we are going to utilise a simple forward Euler simulation scheme for this process $Y^\varepsilon$, on a fixed \emph{space} grid $x_n=x_{n,\delta}:= n\delta$ for the spatial step size $\delta>0$ and $n\in\bb{N}_0$.
We thus define an approximating polygonal process $Y^{\varepsilon, \delta}$ on the points $x_n$ iteratively by
\[
    Y^{\varepsilon, \delta}_{x_{n + 1}} := Y^{\varepsilon, \delta}_{x_n} + \frac{\delta}{f_\varepsilon(Y^{\varepsilon, \delta}_{x_n}, x_n)}
\]
and utilise linear interpolation between.\footnote{So we can write the full process as $ Y^{\varepsilon,\delta}_x = Y^{\varepsilon,\delta}_{x_n} + (x - x_n)/f_\varepsilon(Y^{\varepsilon,\delta}_{x_n}, x_n)$ for $x\in[x_n, x_{n+1})$.}
Now we get the a.s.\ convergence $Y^{\varepsilon, \delta}\to Y^\varepsilon$ uniformly over compacts as $\delta\to0$, like in \citet{mccrickerd_2021} Theorem 2.20.

We similarly define $Z^{\varepsilon,\delta}$ to be the log-linearly interpolating polygonal process between the points
\[
    Z^{\varepsilon,\delta}_{x_n} := s\exp\left((r - q)Y^{\varepsilon, \delta}_{x_n} + \sigma W^\rho_{x_n} - \frac12 \sigma^2 x_n \right).
\]
The fact that this simulation scheme only relies on the Brownian motion $W=(W^0,W^1)$ at the fixed points $x_n$ should not be taken for granted, simplifying implementation and convergence.

So we have defined a parametric polygonal process $(Z^{\varepsilon,\delta}, Y^{\varepsilon,\delta})$, and from this we can recover a polygonal price process $S^{\varepsilon, \delta}:= Z^{\varepsilon,\delta}_{X^{\varepsilon,\delta}}$ where of course $X^{\varepsilon,\delta}:= (Y^{\varepsilon,\delta})^{-1}$.
These processes $Y^{\varepsilon,\delta}$, $Z^{\varepsilon,\delta}$ and $S^{\varepsilon, \delta}$ all a.s.\ converge uniformly over compacts to $Y^{\varepsilon}$, $Z^{\varepsilon}$ and $S^{\varepsilon}$ as $\delta\to 0$, and these in turn converge similarly to $Y$, $Z$ and $S$ as $\varepsilon\to 0$.
So e.g.\ for any continuous and bounded $\Phi,\Psi$ into $\bb{R}$,
\[
\lim_{\varepsilon\to 0}\lim_{\delta\to 0} \left|\Psi\big(\mathbb{E}\big[\Phi\big(S\big)\big]\big) - \Psi\big(\mathbb{E}\big[\Phi\big(S^{\varepsilon,\delta}\big)\big]\big)\right| = 0.
\]
Practically, this means that for any tolerance on derivative prices (or functions thereof), we can always find some $\varepsilon>0$ that will allow us to achieve this tolerance by reducing the step size $\delta$.

\end{appendix}

\end{document}